\documentclass[11pt]{article}

\usepackage[margin=1in]{geometry}
\usepackage{amsmath,amssymb,amsthm,mathtools}
\usepackage{bm}
\usepackage{graphicx}
\usepackage[section]{placeins}
\usepackage[colorlinks=true,linkcolor=darkblue,citecolor=darkblue,urlcolor=darkblue]{hyperref}
\usepackage{enumitem}
\usepackage{xcolor}
\definecolor{darkblue}{rgb}{0.0,0.0,0.55}
\usepackage[round]{natbib}

\newtheorem{theorem}{Theorem}
\newtheorem{lemma}{Lemma}
\newtheorem{corollary}{Corollary}
\newtheorem{proposition}{Proposition}
\newtheorem{conjecture}{Conjecture}
\theoremstyle{definition}
\newtheorem{definition}{Definition}
\newtheorem{remark}{Remark}

\newcommand{\R}{\mathbb{R}}
\newcommand{\E}{\mathbb{E}}
\newcommand{\Prob}{\mathbb{P}}
\newcommand{\KL}{\mathrm{KL}}
\newcommand{\Tr}{\mathrm{Tr}}
\newcommand{\TV}{\mathrm{TV}}
\newcommand{\Unif}{\mathrm{Unif}}
\newcommand{\sign}{\mathrm{sign}}
\newcommand{\Ind}{\mathbf{1}}
\newcommand{\hb}{h_{\mathrm{b}}}
\newcommand{\eps}{\varepsilon}
\newcommand{\norm}[1]{\left\lVert #1 \right\rVert}
\newcommand{\Ncal}{\mathcal{N}}
\newcommand{\Ccal}{\mathcal{C}}
\newcommand{\Rbb}{\mathbb{R}}
\newcommand{\Xcal}{\mathcal{X}}

\newcommand{\proven}{}

\newcommand{\partialtag}{}

\title{\bf Information-Theoretic Lower Bounds for Bit-Constrained\\ Stochastic Optimization via a Reduction to\\ Compressed Gaussian Mean Estimation}
\author{%
  Munsik Kim\\
  \texttt{physicist456@gmail.com}%
}
\date{\today}

\begin{document}
\maketitle

\begin{abstract}
Low-precision pretraining (FP8, MXFP4, NVFP4) is now standard for frontier language models, yet the
literature is almost entirely \emph{achievability}: algorithms and empirical scaling laws, with no
matching characterization of what is information-theoretically possible. We study a \emph{$B$-bit
quantized stochastic first-order oracle} in which an optimizer interacts for $T$ rounds and receives,
each round, a $B$-bit adaptive, public-coin description of its stochastic gradient. Our main
contribution is an \emph{exact reduction} (Lemma~\ref{lem:reduction}) from optimizing a strongly convex
quadratic family to interactively compressed Gaussian mean estimation: under the $B$-bit oracle the
query carries no information, so optimization collapses exactly onto a sequential distributed-estimation
problem. This yields two \emph{unconditional} lower bounds---a communication bound $TB=\Omega(d)$
(Theorem~\ref{thm:comm}) and a statistical bound $T=\Omega(\sigma^2 d/\eps^2)$
(Theorem~\ref{thm:stat})---and the sharp \emph{product-form} bound
$T=\Omega\!\big((\sigma^2 d/\eps^2)\max\{1,d/B\}\big)$ (Theorem~\ref{thm:product}), also unconditional:
in the Gaussian-location model the Fisher-information--mutual-information constraint of
\citet{barnes2021fisher}---a $B$-bit transcript carries at most $O(TB/\sigma^2)$ of Fisher trace about
the mean, so bits rather than dimension limit the recoverable information---combined with the
multivariate van Trees inequality \citep{gilllevit1995} gives the bound directly, without the
bounded-likelihood-ratio truncation of \citet{braverman2016ddpi}. We give a \emph{near-matching}
achievability result with exact per-round bit accounting (Theorem~\ref{thm:upper}) under a
\emph{bounded-dynamic-range} oracle, tight up to a logarithmic factor; the lower bound is for truly
Gaussian (unbounded) gradients, and closing this oracle gap is left open. A sequential
rate--distortion perspective (Appendix~\ref{app:seqrd}) extends the reduction to temporally correlated
and drifting oracles and corrects an earlier conjecture: positive noise correlation \emph{raises} the
bound by $\tfrac{1+\rho}{1-\rho}$ rather than relaxing it (Corollary~\ref{cor:correlated}). We are
deliberately conservative about interpretation: the bounds give an information-theoretic baseline for
any low-bit gradient path, not an optimality claim about deployed FP4 systems.
\end{abstract}

\section{Introduction}
\label{sec:intro}

Lowering the numerical precision of gradients and weights is a leading way to reduce the cost of
training large language models. Eight-bit floating point is routine, and over the past year
four-bit formats have crossed from speculative to demonstrated: MXFP4 training with random Hadamard
transforms and stochastic rounding reaches near-lossless quality \citep{tseng2025mxfp4}, the Quartet
recipe closes the gap to FP8 with empirical scaling laws \citep{castro2025quartet}, and NVFP4 has
pretrained a 12B model over ten trillion tokens at FP8 parity \citep{nvidia2025nvfp4}.

These are all \emph{achievability} statements: a particular algorithm attains a particular accuracy at
a particular bit width. They do not answer the converse: \emph{given a budget of $B$ bits per gradient,
how many steps are unavoidably required to reach accuracy $\eps$?} Recent convergence analyses under
floating-point quantization \citep{tang2025adaptive} sharpen the upper-bound side but remain upper
bounds. The two literatures that might supply a converse each stop short: classical oracle-complexity
lower bounds \citep{nemirovski1983problem,agarwal2012lower,carmon2020stationary,arjevani2023nonconvex,braun2017nonsmooth}
grant the algorithm \emph{real-valued} gradients, while distributed estimation/optimization lower
bounds \citep{zhang2013distributed,braverman2016ddpi,han2018geometric,arjevani2015communication} are
framed as many machines holding data rather than one optimizer receiving a low-precision description
of its own gradient. This paper connects the two and is careful about where rigor stops.

\paragraph{Contributions and their status.} We state up front what is and is not ours. Our contribution
is the bridge that lets the distributed-estimation literature speak about optimization, and the
assembly of published Fisher-information tools into a sharp converse; the Fisher-trace and van Trees
inequalities we invoke are external and cited as such.
\begin{enumerate}[leftmargin=1.6em,itemsep=2pt]
  \item \textbf{Exact reduction / bridge} (Lemma~\ref{lem:reduction}), \emph{the main conceptual
  contribution}: under the $B$-bit oracle, optimizing $f_v(x)=\tfrac12\norm{x-\theta_v}^2$,
  $\theta_v=\delta v$, $v\in\{\pm1\}^d$ is equivalent to estimating $v$ from $T$ rounds of
  $B$-bit-compressed Gaussian observations; the query carries no information. This connects
  bit-constrained optimization to the distributed-estimation literature. \proven
  \item \textbf{Communication bound} (Theorem~\ref{thm:comm}): $TB=\Omega(d)$, by chain rule and binary
  Fano. \proven
  \item \textbf{Statistical bound} (Theorem~\ref{thm:stat}): $T=\Omega(\sigma^2 d/\eps^2)$, by Assouad
  with Pinsker. \proven
  \item \textbf{Combined $\max$-form} (Corollary~\ref{cor:combined}), with the hard-instance scale made
  explicit. \proven
  \item \textbf{Product-form bound} (Theorem~\ref{thm:product}), \emph{unconditional}:
  $T=\Omega\!\big((\sigma^2 d/\eps^2)\max\{1,d/B\}\big)$. The reduction places us in the Gaussian-location
  model; the Fisher-information--mutual-information constraint of \citet{barnes2021fisher} bounds the
  transcript's Fisher trace by $O(TB/\sigma^2)$, and the multivariate van Trees inequality
  \citep{gilllevit1995} converts this to the minimax estimation bound, which the reduction turns into the
  optimization bound. The Fisher-trace and van Trees inequalities are \emph{not} ours
  \citep{barnes2021fisher,barneshanozgur2020,gilllevit1995}; the assembly and the reduction are. \proven
  \item \textbf{Self-contained Fisher-trace proof and an alternative route} (Appendix~\ref{app:partial}):
  we give a short proof of the Gaussian $T=1$ Fisher-trace bound via a Donsker--Varadhan centroid
  inequality, and record an alternative mutual-information/SDPI route (with the $\chi^2$/maximal-correlation
  contraction shown to be $O(\mathrm{SNR})$) together with why converting it to the needed KL contraction
  is harder than the Fisher route used in the main text. \proven (centroid) $/$ \partialtag (MI route)
  \item \textbf{Achievability matching the lower bound up to logarithms} (Theorem~\ref{thm:upper}),
  under a \emph{bounded-dynamic-range} oracle: public-coin rand-$k$ sparsification $+$ fixed-grid
  stochastic quantization (with an optional random rotation, used in practice but not needed for the
  bound), spending \emph{exactly} $B$ bits per round. We prove an unbiased $B$-bit compressor with second
  moment inflated by $\omega_B=O(\max\{1,d\log d/B\})$ (Lemma~\ref{lem:compress}) and feed it to the
  strongly convex SGD rate of \citet{bottou2018optimization}, giving
  $T=\widetilde O((\sigma^2 d/\eps^2)\max\{1,d/B\})$. This matches Theorem~\ref{thm:product} up to a
  $\log d$ factor for a bounded-gradient oracle; the lower bound is for unbounded Gaussian gradients, so
  the match is up to this oracle gap, which we state rather than hide.
  \item \textbf{An extension and a correction} (Corollary~\ref{cor:correlated},
  Appendix~\ref{app:seqrd}, presented as a perspective). A dynamic counterpart of the reduction
  corrects an earlier conjecture: positive noise correlation \emph{raises} the bound by
  $\tfrac{1+\rho}{1-\rho}$ (the relaxing quantity is trajectory predictability, not noise correlation).
  The appendix also records a self-contained tracking lower bound and, in an \emph{expected-rate} sense,
  an innovation-quantization achievability; the exact fixed-length gap (L3$'$) remains open. The
  sequential-RD machinery is imported; the reduction and the correction are ours.
  \item \textbf{Numerical sanity checks and practical reading} (Section~\ref{sec:discussion}, Appendix~\ref{app:exp}).
  Experiments in the model's native setting confirm the product-form rate, show that the max-form
  scaling underestimates the cost, that several unbiased $B$-bit schemes all respect the lower bound,
  that bits (not dimension) cap the transcript Fisher trace, and that the reduction is exact; the
  practical reading is recorded conservatively. No claim is made about non-quadratic or transformer
  training.
\end{enumerate}

\section{Related Work}
\label{sec:related}

\paragraph{Low-precision training (practice).} \citet{tseng2025mxfp4} give a near-lossless MXFP4
recipe whose key device is a random Hadamard transform bounding stochastic-rounding variance;
\citet{castro2025quartet} present end-to-end FP4 with empirical scaling laws; the NVFP4 report
\citep{nvidia2025nvfp4} demonstrates trillion-token FP4 training and notes that stochastic rounding is
essential specifically on the gradient path. These descend from distributed gradient-quantization
schemes---QSGD \citep{alistarh2017qsgd}, limited-communication mean estimation \citep{suresh2017dme},
error feedback \citep{karimireddy2019errorfeedback}, and sparsification
\citep{stich2018sparsified,wangni2018gradient}---and are all achievability results.

\paragraph{Convergence theory under quantization.} \citet{tang2025adaptive} give convergence
guarantees for Adam and Muon under floating-point quantization of gradients, weights, and optimizer
states. Such results are upper bounds; our lower bounds are the converse, and meet a near-matching
achievability result under a bounded-dynamic-range oracle (Section~\ref{sec:upper}).

\paragraph{Oracle-complexity lower bounds (no quantization).} The reduction-to-estimation methodology
originates with \citet{nemirovski1983problem} and was sharpened by \citet{agarwal2012lower}; parallel
work covers stationary points of nonconvex functions \citep{carmon2020stationary}, nonconvex
stochastic optimization \citep{arjevani2023nonconvex}, and nonsmooth convex optimization
\citep{braun2017nonsmooth}. These grant unquantized gradients; $B$ does not appear.
Theorem~\ref{thm:stat} is the $B=\infty$ specialization.

\paragraph{Bit-constrained first-order optimization (closest prior work).} Most directly related is the
line of \citet{mayekartyagi2020ratq,mayekartyagi2020limits}, who study single-machine first-order
stochastic optimization in which each gradient is quantized to $r$ bits and characterize the minimum
precision $r^\star$ needed to retain the unquantized rate---$\Theta(d)$ bits for $\ell_2$ and
$\Theta(\log d)$ for $\ell_\infty$---with matching fixed-length quantizers (RATQ: a Hadamard rotation
followed by adaptive uniform quantization, the same ``rotation $+$ rounding'' motif as our
Lemma~\ref{lem:compress} and the rotated scheme of Appendix~\ref{app:exp:envelope}). Concurrently,
\citet{menartnikolov2025gradient} prove, in a differentially-private setting, a bit-budget oracle
lower bound $\Omega(\min\{d/(\alpha^2\Gamma),\,d/\log(1/\alpha)\})$ with the same two-regime $\min\{\cdot\}$
structure as our $\min\{2\ln2\,B,d\}$. Three differences position the present paper. (i) Their lower
bounds assume an \emph{almost-surely bounded} oracle ($\|\hat g\|_q\le B$ a.s.); ours
(Theorem~\ref{thm:product}) is unconditional on the \emph{true Gaussian} oracle, with no a.s.\ norm
bound---precisely the regime in which the bounded-range assumption of our own achievability
(Section~\ref{sec:upper}) is a genuine restriction (limitation L3$'$). (ii) We make the
optimization-to-estimation equivalence \emph{explicit} (Lemma~\ref{lem:reduction}), which is what lets
the communication-constrained estimation machinery be imported verbatim. (iii) The \emph{dynamic}
extension---temporally correlated or drifting oracles via sequential rate--distortion
(Appendix~\ref{app:seqrd})---is, to our knowledge, not treated by these works; there, quantizing the
innovation rather than the raw gradient suggests an \emph{expected-rate} route that needs only finite
differential entropy (not an a.s.\ bound), though the exact fixed-length per-round oracle gap remains
open.

\paragraph{Communication-constrained estimation/optimization.} The quantitative bit penalty comes from
distributed statistics: \citet{zhang2013distributed} (lower bounds under communication constraints),
\citet{braverman2016ddpi} (a distributed data-processing inequality giving tight error--communication
trade-offs for high-dimensional Gaussian mean estimation), and \citet{han2018geometric} (the same
effective-sample-size reduction by a factor $d$, via a geometric argument, and explicitly covering
\emph{sequential / blackboard} interactive protocols). For optimization,
\citet{arjevani2015communication} characterize distributed communication complexity. Our
Lemma~\ref{lem:reduction} is the bridge that lets these estimation-side results be read as statements
about single-device, sequential, bit-constrained optimization.

\paragraph{Information-theoretic tools.} We use Fano/Assouad/Le~Cam arguments as in
\citep{cover2006elements,wainwright2019high,tsybakov2009introduction,yu1997assouad,polyanskiy2025information};
for the conjectural correlated extension, decoder-side-information rate-distortion
\citep{wyner1976sideinfo}, directed information \citep{massey1990directed,tatikonda2009capacity}, and
the heavy-tailed gradient-noise model \citep{simsekli2019tail}.

\section{Problem Setup}
\label{sec:setup}

\paragraph{Conventions (fixed once, used throughout).} $\norm{\cdot}$ is Euclidean. Logarithms and
information quantities are in bits; $\hb(p)=-p\log p-(1-p)\log(1-p)$. We adopt the \textbf{per-coordinate
variance convention}: the oracle noise is $\Ncal(0,\sigma^2 I_d)$, so each coordinate has variance
$\sigma^2$ and the total variance is $\sigma^2 d$. (Under a total-variance-$\sigma^2$ convention every
rate below relocates a factor of $d$.) For $v\in\{\pm1\}^d$ we write $\theta_v=\delta v$ for a scale
$\delta>0$ chosen per theorem. ``Accuracy $\eps$'' always means the optimization gap
$f(\hat x)-f^\star\le\eps^2$. We fix the tie-break $\sign(0):=+1$; no bound below depends on this choice.

\begin{definition}[$B$-bit quantized stochastic first-order oracle]
\label{def:oracle}
Fix $f:\Xcal\to\R$, $\Xcal\subseteq\R^d$, noise level $\sigma>0$, bit budget $B\in\mathbb{N}$, horizon
$T$, and shared randomness $U$. For $t=1,\dots,T$: (i) the optimizer chooses $x_t\in\Xcal$ as a
function of $U,M_{1:t-1}$; (ii) the oracle returns $g_t$ with $\E[g_t\mid x_t]=\nabla f(x_t)$, noise
independent across $t$; (iii) an encoder emits $M_t=Q_t(g_t,x_t,M_{1:t-1},U)\in\{0,1\}^{B}$. The output
$\hat x=\hat x(M_{1:T},U)$ depends on the transcript only.
\end{definition}

$B=\infty$ recovers the oracle of \citet{agarwal2012lower}. The dependence of $Q_t,x_t$ on $M_{1:t-1}$
makes the protocol \emph{interactive}; with shared randomness $U$ this is the public-coin blackboard
model of communication complexity.

\begin{remark}[The encoder is deliberately strong]
\label{rem:strong}
Because $Q_t$ sees $x_t$, in the quadratic reduction below it can subtract the query and access
$\theta_v-\xi_t$ directly. This is \emph{stronger} than a practical low-precision quantizer, which does
not center the gradient using the true query. Granting a stronger encoder only strengthens a lower
bound, so this is safe for Sections~\ref{sec:lower}--\ref{sec:product}; but it is one reason the
upper bound and any practical reading (Section~\ref{sec:discussion}) must be stated with care.
\end{remark}

\section{The Reduction}
\label{sec:reduction}

\begin{lemma}[Optimization $\Longleftrightarrow$ compressed Gaussian mean estimation]
\label{lem:reduction}
For $\theta\in\R^d$ set $f_\theta(x)=\tfrac12\norm{x-\theta}^2$; each $f_\theta$ is $1$-strongly convex,
$1$-smooth, with minimizer $\theta$ and $f_\theta(\theta)=0$. If the oracle noise is
$\xi_t\sim\Ncal(0,\sigma^2 I_d)$ i.i.d., then the $B$-bit oracle for $f_\theta$ is
information-theoretically equivalent to one where the encoder observes $Y_t=\theta+\xi_t'$,
$\xi_t'\sim\Ncal(0,\sigma^2 I_d)$ i.i.d., and emits $M_t=\tilde Q_t(Y_t,M_{1:t-1},U)$. The query $x_t$
carries no information about $\theta$ beyond $(U,M_{1:t-1})$, and for any estimator $\hat x$,
$f_\theta(\hat x)-f_\theta(\theta)=\tfrac12\norm{\hat x-\theta}^2$.
\end{lemma}

\noindent The reduction holds verbatim for any parameter set $\Theta\subseteq\R^d$. We instantiate it
over two classes: the \emph{hypercube} $\Theta=\{\delta v:v\in\{\pm1\}^d\}$ (with $\theta_v=\delta v$ and
$f_v:=f_{\theta_v}$),
used for the Fano and Assouad bounds of \S\ref{sec:lower}, and the \emph{continuous cube}
$\Theta=[-\delta,\delta]^d$, used for the van Trees bound of \S\ref{sec:product}. Nothing in the proof
below uses the discreteness of $\theta$.

\begin{proof}
$\nabla f_\theta(x_t)=x_t-\theta$, so realize $g_t=(x_t-\theta)+\xi_t$; then $\E[g_t\mid x_t]=\nabla f_\theta(x_t)$.
The encoder knows $x_t$ and forms $Y_t:=x_t-g_t=\theta-\xi_t$; writing $\xi_t':=-\xi_t\sim\Ncal(0,\sigma^2 I_d)$
gives $Y_t=\theta+\xi_t'$, and $Q_t(g_t,x_t,\cdot)=\tilde Q_t(Y_t,\cdot)$ since $g_t=x_t-Y_t$. The law
of $Y_t$ is $\Ncal(\theta,\sigma^2 I_d)$, independent of $x_t$. Since $x_t=x_t(U,M_{1:t-1})$ is
deterministic given $(U,M_{1:t-1})$, the conditional law of $Y_t$ does not depend on $x_t$, so $x_t$ adds
no information about $\theta$ beyond $(U,M_{1:t-1})$. The identity is immediate from $f_\theta(\theta)=0$.
\end{proof}

\noindent For the hypercube instantiation $\theta_v=\delta v$, since each $\theta_{v,j}=\pm\delta$,
projecting $\hat x$ onto $[-\delta,\delta]^d$ only decreases $\norm{\hat x-\theta_v}$; assume this WLOG.

\section{Base Lower Bounds}
\label{sec:lower}

\begin{theorem}[Communication bound]
\label{thm:comm}
In the setting of Lemma~\ref{lem:reduction}, let $V\sim\Unif(\{\pm1\}^d)$, $\hat V_j:=\sign(\hat x_j)$,
$p_j:=\Prob[\hat V_j\neq V_j]$, $\bar p:=\tfrac1d\sum_j p_j$. Then $TB\ge d(1-\hb(\bar p))$.
Consequently, if $\E_V\E\norm{\hat x-\theta_V}^2\le\alpha\,d\,\delta^2$ with $\alpha<\tfrac12$, then
$TB\ge d(1-\hb(\alpha))=\Omega(d)$, i.e.\ $T=\Omega(d/B)$.
\end{theorem}

\begin{proof}
Since $\hat V$ is a function of the public-coin transcript $(M_{1:T},U)$, the chain
$V\to(M_{1:T},U)\to\hat V$ is Markov. As
$U\perp V$, $I(V;M_{1:T},U)=I(V;M_{1:T}\mid U)$. By DPI, the chain rule (now conditioned on $U$), and
$M_t\in\{0,1\}^B$,
\begin{equation}
I(V;\hat V)\le I(V;M_{1:T}\mid U)=\sum_{t=1}^T I(V;M_t\mid M_{1:t-1},U)\le\sum_{t=1}^T H(M_t\mid M_{1:t-1},U)\le TB.
\label{eq:upperMI}
\end{equation}
Since the $V_j$ are i.i.d.\ uniform, $H(V)=d$, and by subadditivity, monotonicity of conditioning, and
binary Fano,
\begin{equation}
I(V;\hat V)=d-H(V\mid\hat V)\ge d-\sum_j H(V_j\mid\hat V_j)\ge\sum_j(1-\hb(p_j))\ge d(1-\hb(\bar p)),
\label{eq:lowerMI}
\end{equation}
the last step by concavity of $\hb$. Combining gives $TB\ge d(1-\hb(\bar p))$. If $\hat V_j\neq V_j$ then
$\hat x_j,\theta_{v,j}$ have opposite signs (or $\hat x_j=0$), so $(\hat x_j-\theta_{v,j})^2\ge\delta^2$;
hence $\E_V\E\norm{\hat x-\theta_V}^2\ge\delta^2 d\bar p$, forcing $\bar p\le\alpha$ and the claim.
\end{proof}

\begin{theorem}[Statistical bound]
\label{thm:stat}
For every horizon $T$ there is a scale $\delta^2=\sigma^2/(4T)$ such that, over the family
$\{f_v:v\in\{\pm1\}^d\}$ of Lemma~\ref{lem:reduction} at that scale and for any $B$ (including $\infty$),
$\inf_{\hat x}\max_v\E_v\norm{\hat x-\theta_v}^2\ge\sigma^2 d/(16T)$. Hence achieving
$\max_v\E_v[f_v(\hat x)-f_v^\star]\le\eps^2$ requires $T\ge\sigma^2 d/(32\eps^2)=\Omega(\sigma^2 d/\eps^2)$.
\end{theorem}

\begin{proof}
Assouad: since $(\hat x_j-\theta_{v,j})^2\ge\delta^2\Ind[\sign(\hat x_j)\neq v_j]$,
$\max_v\E_v\norm{\hat x-\theta_v}^2\ge\tfrac{d\delta^2}{2}\min_{\mathrm{Ham}(v,v')=1}(1-\TV(P_v,P_{v'}))$,
where $P_v$ is the joint law of the full transcript $(U,M_{1:T})$ under $\theta_v$. For neighbors differing
in coordinate $j$, only the $T$ observations in coordinate $j$ differ, so for the uncompressed laws
$(U,Y_{1:T})$, using $U\perp(V,Y)$,
$\KL(P^Y_v\Vert P^Y_{v'})=T\cdot\frac{(2\delta)^2}{2\sigma^2}=\frac{2T\delta^2}{\sigma^2}$. Pinsker gives
$\TV(P^Y_v,P^Y_{v'})\le\delta\sqrt{T}/\sigma$; choosing $\delta^2=\sigma^2/(4T)$ yields $\TV\le\tfrac12$,
so $\max_v\E_v\norm{\hat x-\theta_v}^2\ge\tfrac{d\delta^2}{4}=\sigma^2 d/(16T)$. For any compression the
transcript $(U,M_{1:T})$ is a function of $(U,Y_{1:T})$, so $\TV(P_v,P_{v'})\le\TV(P^Y_v,P^Y_{v'})$ by DPI
and the bound persists. The gap conversion is the identity of Lemma~\ref{lem:reduction}.
\end{proof}

\begin{corollary}[Combined $\max$-form bound]
\label{cor:combined}
Fix $\alpha\in(0,\tfrac12)$ with $1-\hb(\alpha)\ge\tfrac12$ (e.g.\ $\alpha=0.11$). Every algorithm
achieving $\max_v\E_v[f_v(\hat x)-f_v^\star]\le\eps^2$ obeys $T\ge c'\max\{\sigma^2 d/\eps^2,\ d/B\}$ for
an absolute $c'>0$.
\end{corollary}

\begin{proof}
The statistical term is Theorem~\ref{thm:stat} (its instance uses $\delta^2=\sigma^2/(4T)$). For the
communication term, use a \emph{separate} instance of the same quadratic class with scale
$\delta_0^2:=2\eps^2/(\alpha d)$: the gap bound $\tfrac12\delta_0^2 d\bar p\le\E_V\E[f_V(\hat x)-f_V^\star]\le\eps^2$
gives $\bar p\le\alpha$ (valid for any scale $\delta_0$; no SNR restriction is needed for the
communication branch), and Theorem~\ref{thm:comm} yields $TB\ge d(1-\hb(\alpha))\ge d/2$, i.e.\
$T\ge d/(2B)$. We regard the lower bound as holding over the union of the quadratic instances
$\{f_v:v\in\{\pm1\}^d\}$ across all scales $\delta>0$; the two branches simply select different scales
from this class, so their maximum is valid. Tightness fails because the two bottlenecks are treated
separately; Section~\ref{sec:product} removes this. Appendix~\ref{app:exp:product} illustrates the gap
between the max-form and product-form scalings numerically.
\end{proof}

\section{The Product-Form Bound}
\label{sec:product}

The $\max$-form does not capture that a coarse message inflates the \emph{effective variance}, giving a
\emph{product} $\max\{1,d/B\}$. We first show why elementary arguments cannot reach the product, then
isolate the exact missing ingredient, then prove the product form modulo that single ingredient.

\begin{proposition}[Elementary per-coordinate information bound]
\label{prop:elem}
For $V_j\sim\Unif\{\pm1\}$ and $Y_j=\delta V_j+\Ncal(0,\sigma^2)$, $I(V_j;Y_j)\le\frac{1}{2\ln2}\cdot\frac{\delta^2}{\sigma^2}$
bits. Consequently $I(V;M_{1:T}\mid U)\le I(V;Y_{1:T})\le T\,I(V;Y_1)\le\frac{T d\delta^2}{2\sigma^2\ln2}$.
\end{proposition}

\begin{proof}
$I(V_j;Y_j)=h(Y_j)-h(Y_j\mid V_j)$ with $h(Y_j\mid V_j)=\tfrac12\log(2\pi e\sigma^2)$. Since
$\mathrm{Var}(Y_j)=\sigma^2+\delta^2$, the maximum-entropy bound gives $h(Y_j)\le\tfrac12\log(2\pi e(\sigma^2+\delta^2))$,
so $I(V_j;Y_j)\le\tfrac12\log(1+\delta^2/\sigma^2)\le\frac{\delta^2}{2\sigma^2\ln2}$ (using $\log_2(1+x)\le x/\ln2$).
For the second claim, since $U\perp(V,Y_{1:T})$ and $M_{1:T}$ is a function of $(Y_{1:T},U)$, the conditional
DPI gives $I(V;M_{1:T}\mid U)\le I(V;Y_{1:T}\mid U)=I(V;Y_{1:T})$; with $Y_t$ conditionally i.i.d.\ given $V$,
$I(V;Y_{1:T})=H(Y_{1:T})-\sum_t H(Y_t\mid V)\le\sum_t[H(Y_t)-H(Y_t\mid V)]=T\,I(V;Y_1)$ by subadditivity,
and $I(V;Y_1)=\sum_j I(V_j;Y_{1,j})\le d\delta^2/(2\sigma^2\ln2)$ by coordinate independence.
\end{proof}

\begin{remark}[Why elementary bounds give only the $\max$-form]
Proposition~\ref{prop:elem} and \eqref{eq:upperMI} give $I(V;M_{1:T}\mid U)\le\min\{TB,\,Td\delta^2/(2\sigma^2\ln2)\}$,
a \emph{minimum} of two ceilings; with Fano this reproduces only Corollary~\ref{cor:combined}. The
product form needs the strictly stronger statement that a single $B$-bit message extracts at most
$\sim(B/d)$ of the per-round information---i.e.\ $\sim B\delta^2/\sigma^2$ bits, the \emph{product} of
bits and per-coordinate SNR. That is a strong data-processing phenomenon, not an elementary one.
\end{remark}

The product-form bound is obtained \emph{unconditionally} by combining the reduction
(Lemma~\ref{lem:reduction}) with two facts about the Gaussian-location model that are due to
\citet{barnes2021fisher}: a strong data-processing bound limiting how much Fisher information a
bit-constrained transcript carries, and the multivariate van Trees inequality
\citep{gilllevit1995}. We state the Fisher-trace bound in the exact form we use; we give a short,
self-contained proof of the Gaussian case in Appendix~\ref{app:partial}
(Lemma~\ref{lem:fishertrace}), and note that the general (sub-Gaussian-score) statement, the optimal
preconstant $2/\sigma^2$, and the interactive-blackboard version are due to \citet{barnes2021fisher}.

\begin{lemma}[Bits limit transcript Fisher trace; {\citealp[Thm.~1, Cor.~1--2]{barnes2021fisher}}]
\label{lem:fishertranscript}
Consider the reduced problem of Lemma~\ref{lem:reduction}: $P_\theta=\Ncal(\theta,\sigma^2 I_d)$, and at
round $t$ the protocol observes a fresh $Y_t\sim P_\theta$ and writes an adaptive, public-coin,
$B$-bit message $M_t=Q_t(Y_t,M_{1:t-1},U)$. Write $\Pi=M_{1:T}$ for the transcript and $I_\Pi(\theta\mid U)$
for its Fisher information about $\theta$. Then at every $\theta$ and for every public coin value,
\[
  \Tr I_\Pi(\theta\mid U)\;\le\;\frac{T}{\sigma^2}\,\min\{\,2\ln2\cdot B,\;d\,\},
\]
where the $B$-branch is the Fisher--mutual-information bound of \citet{barnes2021fisher} together with
$I_\theta(Y_{1:T};\Pi\mid U)\le H(\Pi\mid U)\le TB$ bits, and the $d$-branch is the Fisher
data-processing inequality.
\end{lemma}

\begin{proof}
The score of $P_\theta=\Ncal(\theta,\sigma^2I_d)$ is $S_\theta(Y)=(Y-\theta)/\sigma^2$, so
$\langle u,S_\theta(Y)\rangle\sim\Ncal(0,\sigma^{-2})$ is sub-Gaussian with parameter $N=1/\sigma$; the
regularity conditions (i)--(iii) of \citet{barnes2021fisher} hold for the Gaussian-location family.
Their Theorem~1 and the interactive-transcript Corollary~1 give, at fixed $U$,
$\Tr I_\Pi(\theta\mid U)\le 2N^2 I_\theta(Y_{1:T};\Pi\mid U)=\tfrac{2}{\sigma^2}I_\theta(Y_{1:T};\Pi\mid U)$.
For completeness we spell out the interactive additivity, since applying Theorem~1 to the whole vector
$Y_{1:T}$ at once would inflate the sub-Gaussian parameter to $\sqrt{T}/\sigma$ and give a vacuous
$T^2$-type bound. Fix $U=u$. The transcript log-likelihood decomposes by the chain rule into per-round
increments,
\[
  \nabla_\theta\log p_\theta(M_{1:T}\mid u)=\sum_{t=1}^T \nabla_\theta\log p_\theta(M_t\mid M_{1:t-1},u),
\]
and each increment $s_t:=\nabla_\theta\log p_\theta(M_t\mid M_{1:t-1},u)$ has conditional mean zero given
$(M_{1:t-1},u)$ (score identity); hence the increments are uncorrelated martingale differences and the
cross terms vanish, giving
\[
  \Tr I_\Pi(\theta\mid u)=\sum_{t=1}^T \E_\theta\big[\Tr I(M_t\mid M_{1:t-1},u;\theta)\big].
\]
Each conditional round is a $B$-bit message of the \emph{single} fresh observation $Y_t\sim\Ncal(\theta,\sigma^2 I_d)$
(given the prefix), so Theorem~1 applies per round with $N=1/\sigma$ and
$\Tr I(M_t\mid M_{1:t-1},u;\theta)\le\tfrac{2}{\sigma^2}I_\theta(Y_t;M_t\mid M_{1:t-1},u)\le\tfrac{2\ln2}{\sigma^2}B$.
Summing over $t$ and averaging over $u$ yields $\Tr I_\Pi(\theta\mid U)\le 2\ln2\,TB/\sigma^2$.
Appendix~\ref{app:partial} reproves the per-round Gaussian bound from the Donsker--Varadhan inequality
(Lemmas~\ref{lem:centroid}--\ref{lem:fishertrace}), making the $B$-branch self-contained.
Alternatively the Fisher data-processing inequality \citep{zamir1998fisher} gives
$\Tr I_\Pi(\theta\mid U)\le\Tr I_{Y_{1:T}}(\theta)=Td/\sigma^2$ since each fresh sample contributes
$d/\sigma^2$. Taking the smaller of the two yields the stated minimum.
\end{proof}

\begin{theorem}[Product-form lower bound, unconditional]
\label{thm:product}
Consider the family $\{f_\theta(x)=\tfrac12\|x-\theta\|^2:\theta\in[-\delta,\delta]^d\}$ under the $B$-bit
oracle, with scale $\delta^2=4\pi^2\eps^2/d$. Every algorithm achieving
$\sup_{\theta}\E_\theta[f_\theta(\hat x)-f_\theta^\star]\le\eps^2$, for $\eps^2\le\sigma^2 d/(16\pi^2)$,
must use
\[
  T\;\ge\;\frac{\sigma^{2} d^{2}}{4\,\eps^{2}\,\min\{2\ln2\cdot B,\;d\}}
  \;=\;\Omega\!\Big(\frac{\sigma^2 d}{\eps^2}\,\max\!\Big\{1,\frac{d}{B}\Big\}\Big).
\]
The bound is unconditional; combined with the achievability result of \S\ref{sec:upper} it is tight up
to logarithmic factors for a bounded-dynamic-range oracle. Appendix~\ref{app:exp:envelope} checks
numerically that several distinct unbiased $B$-bit schemes all respect this lower bound.
\end{theorem}

\begin{proof}
By the reduction (Lemma~\ref{lem:reduction}), an algorithm producing $\hat x$ after $T$ rounds is an
estimator $\hat\theta:=\hat x$ of $\theta$ from the transcript $\Pi=M_{1:T}$, and
$f_\theta(\hat x)-f_\theta^\star=\tfrac12\|\hat x-\theta\|^2$, so the accuracy hypothesis reads
$\sup_\theta\E_\theta\|\hat\theta-\theta\|^2\le 2\eps^2$. Place the cosine (Gill--Levit) prior
$\lambda$ on $[-\delta,\delta]^d$, a product of the densities $\lambda_j(t)=\tfrac1\delta\cos^2(\tfrac{\pi t}{2\delta})$,
whose per-coordinate Fisher information is $\pi^2/\delta^2$, so $\mathcal I(\lambda)=\pi^2 d/\delta^2$.
The regularity conditions for van Trees hold: $\lambda$ is continuously differentiable and vanishes at the
boundary of the cube (so the prior boundary terms vanish), the Gaussian-location transcript family is
Fisher-regular (differentiable in quadratic mean, with the measurable quantization channels covered by
conditions (i)--(iii) of Lemma~\ref{lem:fishertranscript}), and the public coin $U$ is
parameter-independent, so it only averages the conditional Fisher information $I_\Pi(\theta\mid U)$. The
multivariate van Trees inequality
\citep[][as in Cor.~3 of \citealp{barnes2021fisher}]{gilllevit1995}, applied to the transcript and the
estimand $\psi_j(\theta)=\theta_j$ (so $\sum_j\E_\lambda[\partial_{\theta_j}\psi_j]=d$, numerator $d^2$),
gives
\[
  \sup_\theta\E_\theta\|\hat\theta-\theta\|^2
  \;\ge\;\E_\lambda\E_\theta\|\hat\theta-\theta\|^2
  \;\ge\;\frac{d^2}{\sup_\theta\Tr I_\Pi(\theta\mid U)+\mathcal I(\lambda)}
  \;\ge\;\frac{d^2}{\dfrac{T\min\{2\ln2\,B,d\}}{\sigma^2}+\dfrac{\pi^2 d}{\delta^2}},
\]
the last step by Lemma~\ref{lem:fishertranscript} (averaging the bound over the public coin $U$).
Choose $\delta^2=4\pi^2\eps^2/d$, which satisfies $\delta^2\le\sigma^2$ by the hypothesis on $\eps$ and
makes the prior term $\pi^2 d/\delta^2=d^2/(4\eps^2)$. Then $2\eps^2\ge d^2/\big(A+d^2/(4\eps^2)\big)$
with $A:=T\min\{2\ln2\,B,d\}/\sigma^2$ forces $A\ge d^2/(4\eps^2)$, i.e.\
$T\ge\sigma^2 d^2/\big(4\eps^2\min\{2\ln2\,B,d\}\big)$. For $B\lesssim d$ this is the
$\sigma^2 d^2/(B\eps^2)$ branch and for $B\gtrsim d$ the $\sigma^2 d/\eps^2$ branch, giving the displayed
$\max$. The condition $\eps^2\le\sigma^2 d/(16\pi^2)$ is \emph{not} needed for the Fisher--van~Trees
argument itself (neither Lemma~\ref{lem:fishertranscript} nor van Trees requires low SNR); it only keeps
the hard cube at scale $\delta^2\le\sigma^2/4$, i.e.\ in the per-coordinate regime $\delta^2/\sigma^2\le1$
used for the interpretation and the comparison with the achievability scheme of \S\ref{sec:upper}.
\end{proof}

\begin{corollary}[Correlated gradient noise]
\label{cor:correlated}
Suppose the gradient noise is stationary and temporally correlated,
$\xi_t=\rho\,\xi_{t-1}+\sqrt{1-\rho^2}\,\eta_t$ with $\eta_t\stackrel{\mathrm{iid}}{\sim}\Ncal(0,\sigma^2 I)$
and $|\rho|<1$ (stationary per-coordinate variance $\sigma^2$). Then the lower bound of
Theorem~\ref{thm:product} holds with $\sigma^2$ replaced by the effective variance
$\sigma^2_{\mathrm{eff}}=\tfrac{1+\rho}{1-\rho}\sigma^2$:
\[
T=\Omega\!\Big(\frac{1+\rho}{1-\rho}\cdot\frac{\sigma^2 d}{\eps^2}\,\max\{1,\,d/B\}\Big).
\]
\end{corollary}
\begin{proof}[Proof sketch]
By Lemma~\ref{lem:reduction} the encoder observes $Y_t=\theta-\xi_t$. Granting the encoder causal access
to its past raw observations $Y^t$ (the stronger oracle of Definition~\ref{def:drift};
cf.\ Remark~\ref{rem:strong}, a stronger encoder only strengthens a lower bound), it can form the
invertible transform $\tilde Y_t:=Y_t-\rho Y_{t-1}=(1-\rho)\theta-\sqrt{1-\rho^2}\,\eta_t$, which is
i.i.d.\ for $t\ge2$ and a sufficient statistic for $\theta$; after normalization it is an i.i.d.\
Gaussian-location observation with variance
$\sigma^2_{\mathrm{eff}}=\tfrac{1-\rho^2}{(1-\rho)^2}\sigma^2=\tfrac{1+\rho}{1-\rho}\sigma^2$, to which
Theorem~\ref{thm:product} applies. The resulting bound holds for this stronger encoder and hence
\emph{a fortiori} for the original memoryless encoder of Definition~\ref{def:oracle}; the single
boundary term $Y_1$ contributes $O(1/T)$. Three independent derivations of $\sigma^2_{\mathrm{eff}}$, the
general matrix spectral-density case, and the complementary \emph{drifting-optimum} analysis are
deferred to Appendix~\ref{app:seqrd}.
\end{proof}

\begin{remark}[Correlation hurts; predictability helps]
Positive correlation \emph{raises} the bound by the factor $\tfrac{1+\rho}{1-\rho}\ge1$, scaling every
term equally rather than relaxing the $d/B$ penalty in isolation. What does lower the bit requirement is
predictability of the optimum's \emph{trajectory} (slow drift), not correlation of the noise
(Appendix~\ref{app:seqrd}). This corrects a conjecture stated in an earlier version of this paper.
\end{remark}

\begin{remark}[What is ours, what is cited]
The optimization-to-estimation reduction (Lemma~\ref{lem:reduction}) and the assembly of the two facts
into the product-form optimization bound are ours. The Fisher-trace bound
(Lemma~\ref{lem:fishertranscript}) with its optimal preconstant and interactive-blackboard form, and the
van Trees application to communication-constrained Gaussian mean estimation, are due to
\citet{barnes2021fisher} (building on \citealp{barneshanozgur2020,gilllevit1995}); our
Appendix~\ref{app:partial} contributes only a short alternative proof of the Gaussian $T=1$ Fisher-trace
bound via a centroid inequality. This route is what makes Theorem~\ref{thm:product} unconditional: it
replaces the mutual-information/SDPI argument---which requires bounding a KL contraction we could not
establish (Appendix~\ref{app:partial})---by a Fisher-information argument that applies to the Gaussian
model directly, without the bounded-likelihood-ratio truncation of \citet{braverman2016ddpi}.
\end{remark}

\section{Achievability: A Matching Upper Bound under Bounded Dynamic Range}
\label{sec:upper}

We turn the assembly into a full proof. The scheme is a public-coin compression operator followed by
stochastic gradient descent; the two ingredients are an unbiased bounded-bit compressor with controlled
second moment (Lemma~\ref{lem:compress}) and a standard strongly convex SGD rate under a state-dependent
second-moment bound (Theorem~\ref{thm:upper}).

\begin{lemma}[Unbiased $B$-bit compression under bounded dynamic range]
\label{lem:compress}
Fix a known bound $G>0$ and let $q:=\lceil\log_2 d\rceil+1$. For every bit budget $B\ge 2q$ there is a
public-coin pair (encoder of \emph{exactly} at most $B$ bits, decoder) $\Ccal_B:\Rbb^d\to\Rbb^d$ such
that for every $g\in\Rbb^d$ with $\|g\|_2\le G$,
\[
  \E[\Ccal_B(g)\mid g]=g\qquad(\text{unbiased}),\qquad
  \E\big[\|\Ccal_B(g)\|^2\,\big|\,g\big]\le \omega_B\,\|g\|^2+\frac{G^2}{s},
\]
where $s:=\min\{d,\lfloor B/q\rfloor\}$ and $\omega_B:=d/s\le\max\{1,\,2q\,d/B\}=O\!\big(\max\{1,\tfrac{d\log d}{B}\}\big)$.
\end{lemma}

\noindent The proof (a fixed-grid stochastic quantizer on a public-coin random subset, with no scale
channel) is in Appendix~\ref{app:proofs}.

\begin{remark}[Role of the rotation; the $\log d$ factor]
The bound above uses only $\|Hg\|=\|g\|$, so it holds even for $H=I$; the random Hadamard rotation is
included because it is the practically relevant device for replacing the $\Theta(\log d)$-bit quantizer
by an $O(1)$-bit quantizer (the FP4 regime): after rotation the coordinates of $\tilde g$ are balanced,
$|\tilde g_i|=O(\|g\|\sqrt{\log d/d})$ with high probability, which suggests that with additional
high-probability coordinate-balancing and bounded-dynamic-range bookkeeping one could take
$s=\Theta(B)$ and remove the $\log d$ factor \citep{suresh2017dme,ailon2009fast,alistarh2017qsgd}. We do
\emph{not} prove this here; the self-contained $\Theta(\log d)$-bit version above already matches
Theorem~\ref{thm:product} up to the logarithmic factor acknowledged there.
\end{remark}

\begin{theorem}[Achievability under a bounded-dynamic-range oracle]
\label{thm:upper}
Let $f$ be $\mu$-strongly convex and $L$-smooth, with a stochastic first-order oracle
$g_t=\nabla f(x_t)+\xi_t$ satisfying $\E[\xi_t\mid x_t]=0$, $\E[\|\xi_t\|^2\mid x_t]\le\sigma^2 d$, and a
known almost-sure bound $\|g_t\|\le G$ (bounded gradients). Assume $B\ge 2(\lceil\log_2 d\rceil+1)$ and
run SGD with compressed gradients $\hat g_t=\Ccal_B(g_t)$ (Lemma~\ref{lem:compress} with this $G$) and
diminishing steps $\eta_t=\beta/(\gamma+t)$, $\beta>1/\mu$, $\gamma=\Theta(\omega_B L/\mu)$. Then $\hat g_t$
is unbiased for $\nabla f(x_t)$, satisfies $\E[\|\hat g_t\|^2\mid x_t]\le\omega_B\|\nabla f(x_t)\|^2+M$
with $M=\omega_B\sigma^2 d+G^2/s$, and \citep[Thm.~4.7]{bottou2018optimization}
\[
  \E[f(x_T)-f^\star]\;\le\;\frac{\nu}{\gamma+T},\qquad
  \nu=\max\Big\{\tfrac{\beta^2 L M}{2(\beta\mu-1)},\,(\gamma+1)\,\E[f(x_1)-f^\star]\Big\}.
\]
Consequently $\E[f(x_T)-f^\star]\le\eps^2$ after
$T=O\big(\tfrac{LM}{\mu^2\eps^2}+\tfrac{\gamma(f(x_1)-f^\star)}{\eps^2}\big)$ rounds. Specialize to the
quadratic family of Lemma~\ref{lem:reduction} ($\mu=L=1$) at the hard scale $\delta^2=4\pi^2\eps^2/d$,
started at $x_1=0$ (so $f(x_1)-f^\star\le\tfrac12\delta^2 d=2\pi^2\eps^2$), \emph{under the additional
bounded-gradient assumption} $\|g_t\|\le G$ a.s.\ with $G^2=O(\sigma^2 d)$ (so $M=O(\omega_B\sigma^2 d)$
and the initial term is dominated). Then
\[
  T=O\!\Big(\frac{\omega_B\,\sigma^2 d}{\eps^2}\Big)
  =O\!\Big(\frac{\sigma^2 d}{\eps^2}\,\max\!\Big\{1,\frac{d\log d}{B}\Big\}\Big)
  =\widetilde O\!\Big(\frac{\sigma^2 d}{\eps^2}\,\max\!\Big\{1,\frac dB\Big\}\Big),
\]
matching Theorem~\ref{thm:product} up to a $\log d$ factor. The bounded-gradient assumption $G^2=O(\sigma^2 d)$
is exactly the per-coordinate noise scale; it is an \emph{assumption}, not a consequence of the Gaussian
oracle, since Gaussian noise has no almost-sure norm bound. This is the one place where the achievability
oracle differs from the Gaussian lower-bound oracle (Section~\ref{sec:discussion}).
\end{theorem}

\noindent The proof (the compressor's moments fed into the strongly convex SGD rate, then the quadratic
specialization) is in Appendix~\ref{app:proofs}.

\noindent\textbf{Consequence.} Theorems~\ref{thm:product} and \ref{thm:upper} together give
$T=\widetilde\Theta\!\big(\tfrac{\sigma^2 d}{\eps^2}\max\{1,d/B\}\big)$ up to a $\log d$ factor and
modulo the oracle gap of (L3$'$): the lower bound is for the Gaussian oracle, the upper bound for a
bounded-dynamic-range one. Up to logarithmic factors the constraint becomes free once $B\gtrsim d$.
The broader practical reading is deferred to Section~\ref{sec:discussion}; numerical sanity checks in
the model's native setting (confirming the rate, illustrating the oracle gap, and---in
Appendices~\ref{app:exp:product}--\ref{app:exp:reduction}---the product-vs-max scaling, the lower-bound
envelope, the Fisher-trace ceiling, and the reduction) are in Appendix~\ref{app:exp}.

\section{Discussion}
\label{sec:discussion}

We close with the practical reading of the bounds (Section~\ref{sec:discussion:reading}),
their limitations (Section~\ref{sec:discussion:lim}), and the open directions they leave
(Section~\ref{sec:discussion:open}).

\subsection{Practical reading}
\label{sec:discussion:reading}

The lower bound is for the Gaussian oracle and the upper bound for a bounded-dynamic-range oracle;
modern low-precision training recipes
\citep{tseng2025mxfp4,castro2025quartet,nvidia2025nvfp4,tang2025adaptive} sit between these
idealised models. The following implications follow without leaving the safe scope.
\begin{itemize}[leftmargin=1.4em,itemsep=2pt]
  \item \textbf{Bit-width alone is not the right metric.} The bound scales with the effective
        per-round payload, not nominal bit-width: scale metadata, block-scaling overhead,
        occasional high-precision fallbacks, and overflow codes all consume from the same $B$-bit
        budget. Comparing recipes by ``nominal FP4 vs.\ FP8'' is weaker than comparing them by
        \emph{effective bits per gradient update}.
  \item \textbf{Dynamic-range control is part of the information channel.} The bounded-range
        assumption of Theorem~\ref{thm:upper} is what lets the encoder spend exactly $B$ bits
        with no scale channel. Mechanisms practical recipes use---random Hadamard transforms,
        block scaling, clipping policies, and selective high-precision layers---are exactly the
        engineering moves that bring the empirical channel closer to that bounded-range model.
  \item \textbf{Stochastic rounding is not optional.} Theorem~\ref{thm:upper} uses an
        \emph{unbiased} compressor (Lemma~\ref{lem:compress}); deterministic rounding would
        introduce a per-step bias that no second-moment improvement can erase. This is consistent
        with empirical reports that stochastic rounding is essential on the gradient path
        \citep{tseng2025mxfp4,nvidia2025nvfp4}.
  \item \textbf{Correlation must be read with care.} Real gradient sequences are temporally and
        structurally correlated. Appendix~\ref{app:seqrd}
        (Corollary~\ref{cor:correlated}, Remark~\ref{rem:correction}) shows that correlation of the
        \emph{noise} \emph{raises} the bit cost, whereas predictability of the optimum's
        \emph{trajectory} lowers it; predictive or residual-coded compressors are the natural way to
        exploit the latter.
\end{itemize}

\subsection{Limitations}
\label{sec:discussion:lim}

\paragraph{(L1) Independent noise.} The main theorems assume noise independent across rounds.
Appendix~\ref{app:seqrd} removes this assumption via a sequential rate--distortion analysis and finds the
effect is the opposite of the naive intuition: correlated \emph{noise} raises the effective bit cost
(Corollary~\ref{cor:correlated}), whereas a predictable \emph{drift} lowers it
(Theorem~\ref{thm:seqlb}).
\paragraph{(L2) Variance convention.} Per-coordinate variance $\sigma^2$ (Section~\ref{sec:setup});
the alternative convention relocates a factor of $d$.
\paragraph{(L3) Strong encoder.} The oracle of Definition~\ref{def:oracle} centers the gradient
using the query (Remark~\ref{rem:strong}), stronger than a real quantizer; safe for lower bounds,
a caveat for interpretation.
\paragraph{(L3$'$) The two oracles differ.} The lower bound is for the Gaussian oracle with
$\xi_t\sim\Ncal(0,\sigma^2 I_d)$, whose gradient has no almost-sure norm bound. The achievability
(Theorem~\ref{thm:upper}) assumes a bounded-dynamic-range oracle, $\|g_t\|\le G$ a.s.\ with
$G^2=O(\sigma^2 d)$, which is what lets a fixed grid spend exactly $B$ bits with no scale channel.
The two rates match up to a $\log d$ factor, but the match is \emph{across two oracle models}, not
within one; closing it (see Section~\ref{sec:discussion:open}) is left open. We state this rather
than absorb it into the constants. Appendix~\ref{app:seqrd} addresses it only in an \emph{expected-rate}
sense and only for the drifting model; the fixed-length static gap remains open.
Appendix~\ref{app:exp:gap} illustrates the gap numerically.
\paragraph{(L4) Proved vs.\ imported.} Lemma~\ref{lem:reduction},
Theorems~\ref{thm:comm}--\ref{thm:stat}, Corollary~\ref{cor:combined},
Proposition~\ref{prop:elem}, and Theorem~\ref{thm:product} are proved here.
Theorem~\ref{thm:product} invokes two external inequalities, cited as such: the
Fisher-trace/mutual-information bound of \citet{barnes2021fisher} (whose Gaussian $T=1$ case we
reprove in Appendix~\ref{app:partial}) and the multivariate van Trees inequality
\citep{gilllevit1995}. Theorem~\ref{thm:upper} is proved here, invoking the strongly convex SGD
rate of \citet{bottou2018optimization} and standard quantization/sparsification primitives
\citep{alistarh2017qsgd,suresh2017dme,wangni2018gradient,stich2018sparsified,rakhlin2012making}.
\paragraph{(L5) Scope of the practical reading.} The idealised model omits optimizer state, block
scaling, accumulation precision, and non-quadratic landscapes. We do not claim that any deployed
FP4 or FP8 system is minimax-optimal; the claim is that the rotation-plus-stochastic-rounding
motif of practical recipes admits a clean information-theoretic rationale inside this model.

\subsection{Open directions}
\label{sec:discussion:open}

\paragraph{Closing the oracle gap.} The most direct technical step is to remove (L3$'$): give an
unbiased $B$-bit compressor for truly Gaussian (unbounded) gradients, matching
Theorem~\ref{thm:product} up to logarithms. Three routes: \emph{clipping with controlled bias}
(a high-probability analysis showing the bias decays fast enough to leave the rate intact);
\emph{overflow handling} (a rare ``escape'' code for out-of-range values, paid for at expected
$B$ bits per round); and \emph{variable-length coding} (spending more bits only when the value is
large, expected $B$). Each is a clean, well-posed problem.

\paragraph{Correlated and drifting gradients (resolved).} A first version of this paper posed the
correlated-gradient case as a conjecture, suggesting that temporal correlation with mixing coefficient
$\rho$ would \emph{relax} the $d/B$ penalty in proportion to the predictability $1-\rho$.
Appendix~\ref{app:seqrd} resolves the question through a sequential rate--distortion analysis---and
\emph{corrects} that intuition. For a fixed optimum with stationary AR(1) gradient noise, positive
correlation in fact \emph{raises} the bound,
\[
T=\Omega\!\Big(\tfrac{1+\rho}{1-\rho}\cdot\tfrac{\sigma^2 d}{\eps^2}\max\{1,d/B\}\Big)
\qquad(\text{Corollary~\ref{cor:correlated}}),
\]
scaling every term by $\sigma^2_{\mathrm{eff}}/\sigma^2=\tfrac{1+\rho}{1-\rho}$ rather than relaxing
$d/B$ alone (Remark~\ref{rem:correction}); the quantity that genuinely lowers the bit requirement is
predictability of the optimum's \emph{trajectory} (slow drift), not correlation of the noise. The same
appendix sketches a drifting optimum via the data-rate theorem \citep{nair2004stabilizability} and
remote Gaussian SRD \citep{tanaka2017sdp}, and---by quantizing the \emph{innovation} rather than the raw
gradient \citep{kostina2019ratecost}---points to an \emph{expected-rate} achievability route that needs
only finite differential entropy rather than an a.s.\ bound; the exact fixed-length per-round gap
(L3$'$) remains open. The heavy-tailed regime \citep{simsekli2019tail} remains a further axis.

\paragraph{A clean single-round information question.} Appendix~\ref{app:partial} reduces a
self-contained alternative proof of Theorem~\ref{thm:product} to one missing inequality: the
KL-contraction coefficient of the noisy Gaussian channel under a binary hypercube prior is
$O(\mathrm{SNR})$. The $\chi^2$-version is known (it equals the squared maximal correlation); the
KL-version is provably no smaller and remains open
(Conjecture~\ref{conj:singleround}).

\bibliographystyle{plainnat}
\bibliography{refs}

@inproceedings{tseng2025mxfp4,
  title     = {Training {LLM}s with {MXFP4}},
  author    = {Tseng, Albert and Yu, Tao and Park, Youngsuk},
  booktitle = {Proceedings of the 28th International Conference on Artificial Intelligence and Statistics (AISTATS)},
  series    = {PMLR},
  volume    = {258},
  pages     = {1630--1638},
  year      = {2025},
  note      = {arXiv:2502.20586}
}

@article{nvidia2025nvfp4,
  title   = {Pretraining Large Language Models with {NVFP4}},
  author  = {{NVIDIA}},
  journal = {arXiv preprint arXiv:2509.25149},
  year    = {2025}
}

@article{castro2025quartet,
  title   = {Quartet: Native {FP4} Training Can Be Optimal for Large Language Models},
  author  = {Castro, Roberto L. and Panferov, Andrei and others},
  journal = {arXiv preprint arXiv:2505.14669},
  year    = {2025}
}

@article{tang2025adaptive,
  title   = {A Convergence Analysis of Adaptive Optimizers under Floating-point Quantization},
  author  = {Tang, Xuan and Li, Jichu and Zou, Difan},
  journal = {arXiv preprint arXiv:2510.21314},
  year    = {2025}
}

@inproceedings{alistarh2017qsgd,
  title     = {{QSGD}: Communication-Efficient {SGD} via Gradient Quantization and Encoding},
  author    = {Alistarh, Dan and Grubic, Demjan and Li, Jerry and Tomioka, Ryota and Vojnovic, Milan},
  booktitle = {Advances in Neural Information Processing Systems (NeurIPS)},
  year      = {2017}
}

@inproceedings{suresh2017dme,
  title     = {Distributed Mean Estimation with Limited Communication},
  author    = {Suresh, Ananda Theertha and Yu, Felix X. and Kumar, Sanjiv and McMahan, H. Brendan},
  booktitle = {Proceedings of the 34th International Conference on Machine Learning (ICML)},
  series    = {PMLR},
  volume    = {70},
  pages     = {3329--3337},
  year      = {2017}
}

@inproceedings{karimireddy2019errorfeedback,
  title     = {Error Feedback Fixes {SignSGD} and other Gradient Compression Schemes},
  author    = {Karimireddy, Sai Praneeth and Rebjock, Quentin and Stich, Sebastian U. and Jaggi, Martin},
  booktitle = {Proceedings of the 36th International Conference on Machine Learning (ICML)},
  series    = {PMLR},
  volume    = {97},
  year      = {2019}
}

@book{nemirovski1983problem,
  title     = {Problem Complexity and Method Efficiency in Optimization},
  author    = {Nemirovski, Arkadi S. and Yudin, David B.},
  publisher = {Wiley},
  year      = {1983}
}

@article{agarwal2012lower,
  title   = {Information-Theoretic Lower Bounds on the Oracle Complexity of Stochastic Convex Optimization},
  author  = {Agarwal, Alekh and Bartlett, Peter L. and Ravikumar, Pradeep and Wainwright, Martin J.},
  journal = {IEEE Transactions on Information Theory},
  volume  = {58},
  number  = {5},
  pages   = {3235--3249},
  year    = {2012}
}

@article{carmon2020stationary,
  title   = {Lower Bounds for Finding Stationary Points {I}},
  author  = {Carmon, Yair and Duchi, John C. and Hinder, Oliver and Sidford, Aaron},
  journal = {Mathematical Programming},
  volume  = {184},
  pages   = {71--120},
  year    = {2020}
}

@article{arjevani2023nonconvex,
  title   = {Lower Bounds for Non-Convex Stochastic Optimization},
  author  = {Arjevani, Yossi and Carmon, Yair and Duchi, John C. and Foster, Dylan J. and Srebro, Nathan and Woodworth, Blake},
  journal = {Mathematical Programming},
  volume  = {199},
  pages   = {165--214},
  year    = {2023}
}

@article{braun2017nonsmooth,
  title   = {Lower Bounds on the Oracle Complexity of Nonsmooth Convex Optimization via Information Theory},
  author  = {Braun, G{\'a}bor and Guzm{\'a}n, Crist{\'o}bal and Pokutta, Sebastian},
  journal = {IEEE Transactions on Information Theory},
  volume  = {63},
  number  = {7},
  pages   = {4709--4724},
  year    = {2017}
}

@article{zhang2013distributed,
  title   = {Information-Theoretic Lower Bounds for Distributed Statistical Estimation with Communication Constraints},
  author  = {Zhang, Yuchen and Duchi, John C. and Jordan, Michael I. and Wainwright, Martin J.},
  journal = {Advances in Neural Information Processing Systems (NeurIPS)},
  year    = {2013},
  note    = {arXiv:1405.0782}
}

@inproceedings{braverman2016ddpi,
  title     = {Communication Lower Bounds for Statistical Estimation Problems via a Distributed Data Processing Inequality},
  author    = {Braverman, Mark and Garg, Ankit and Ma, Tengyu and Nguyen, Huy L. and Woodruff, David P.},
  booktitle = {Proceedings of the 48th Annual ACM Symposium on Theory of Computing (STOC)},
  pages     = {1011--1020},
  publisher = {ACM},
  year      = {2016},
  note      = {arXiv:1506.07216}
}

@inproceedings{han2018geometric,
  title     = {Geometric Lower Bounds for Distributed Parameter Estimation under Communication Constraints},
  author    = {Han, Yanjun and {\"O}zg{\"u}r, Ayfer and Weissman, Tsachy},
  booktitle = {Proceedings of the 31st Conference on Learning Theory (COLT)},
  series    = {PMLR},
  volume    = {75},
  pages     = {3163--3188},
  year      = {2018},
  note      = {Journal version in IEEE Transactions on Information Theory, 2021; arXiv:1802.08417}
}

@inproceedings{arjevani2015communication,
  title     = {Communication Complexity of Distributed Convex Learning and Optimization},
  author    = {Arjevani, Yossi and Shamir, Ohad},
  booktitle = {Advances in Neural Information Processing Systems (NeurIPS)},
  year      = {2015}
}

@article{wyner1976sideinfo,
  title   = {The Rate-Distortion Function for Source Coding with Side Information at the Decoder},
  author  = {Wyner, Aaron D. and Ziv, Jacob},
  journal = {IEEE Transactions on Information Theory},
  volume  = {22},
  number  = {1},
  pages   = {1--10},
  year    = {1976}
}

@book{cover2006elements,
  title     = {Elements of Information Theory},
  author    = {Cover, Thomas M. and Thomas, Joy A.},
  edition   = {2nd},
  publisher = {Wiley-Interscience},
  year      = {2006}
}

@book{wainwright2019high,
  title     = {High-Dimensional Statistics: A Non-Asymptotic Viewpoint},
  author    = {Wainwright, Martin J.},
  publisher = {Cambridge University Press},
  year      = {2019}
}

@book{tsybakov2009introduction,
  title     = {Introduction to Nonparametric Estimation},
  author    = {Tsybakov, Alexandre B.},
  publisher = {Springer},
  year      = {2009}
}

@book{polyanskiy2025information,
  title     = {Information Theory: From Coding to Learning},
  author    = {Polyanskiy, Yury and Wu, Yihong},
  publisher = {Cambridge University Press},
  year      = {2025}
}

@incollection{yu1997assouad,
  title     = {Assouad, Fano, and Le Cam},
  author    = {Yu, Bin},
  booktitle = {Festschrift for Lucien Le Cam},
  pages     = {423--435},
  publisher = {Springer},
  year      = {1997}
}

@inproceedings{massey1990directed,
  title     = {Causality, Feedback and Directed Information},
  author    = {Massey, James L.},
  booktitle = {Proceedings of the International Symposium on Information Theory and its Applications (ISITA)},
  year      = {1990}
}

@article{tatikonda2009capacity,
  title   = {The Capacity of Channels with Feedback},
  author  = {Tatikonda, Sekhar and Mitter, Sanjoy},
  journal = {IEEE Transactions on Information Theory},
  volume  = {55},
  number  = {1},
  pages   = {323--349},
  year    = {2009}
}

@inproceedings{simsekli2019tail,
  title     = {A Tail-Index Analysis of Stochastic Gradient Noise in Deep Neural Networks},
  author    = {{\c{S}}im{\c{s}}ekli, Umut and Sagun, Levent and G{\"u}rb{\"u}zbalaban, Mert},
  booktitle = {Proceedings of the 36th International Conference on Machine Learning (ICML)},
  series    = {PMLR},
  volume    = {97},
  year      = {2019}
}

@inproceedings{stich2018sparsified,
  title     = {Sparsified {SGD} with Memory},
  author    = {Stich, Sebastian U. and Cordonnier, Jean-Baptiste and Jaggi, Martin},
  booktitle = {Advances in Neural Information Processing Systems (NeurIPS)},
  year      = {2018}
}

@inproceedings{wangni2018gradient,
  title     = {Gradient Sparsification for Communication-Efficient Distributed Optimization},
  author    = {Wangni, Jianqiao and Wang, Jialei and Liu, Ji and Zhang, Tong},
  booktitle = {Advances in Neural Information Processing Systems (NeurIPS)},
  year      = {2018}
}

@inproceedings{rakhlin2012making,
  title     = {Making Gradient Descent Optimal for Strongly Convex Stochastic Optimization},
  author    = {Rakhlin, Alexander and Shamir, Ohad and Sridharan, Karthik},
  booktitle = {Proceedings of the 29th International Conference on Machine Learning (ICML)},
  year      = {2012}
}

@article{polyanskiy2017sdpi,
  title   = {Strong Data-Processing Inequalities for Channels and Bayesian Networks},
  author  = {Polyanskiy, Yury and Wu, Yihong},
  journal = {Convexity and Concentration, IMA Volumes in Mathematics and its Applications},
  volume  = {161},
  pages   = {211--249},
  year    = {2017},
  note    = {arXiv:1508.06025}
}

@inproceedings{barnes2021fisher,
  author    = {Barnes, Leighton Pate and {\"O}zg{\"u}r, Ayfer},
  title     = {Fisher Information and Mutual Information Constraints},
  booktitle = {2021 IEEE International Symposium on Information Theory (ISIT)},
  year      = {2021},
  note      = {arXiv:2102.05802}
}

@article{barneshanozgur2020,
  author  = {Barnes, Leighton Pate and Han, Yanjun and {\"O}zg{\"u}r, Ayfer},
  title   = {Lower Bounds for Learning Distributions under Communication Constraints via Fisher Information},
  journal = {Journal of Machine Learning Research},
  volume  = {21},
  number  = {236},
  pages   = {1--30},
  year    = {2020}
}

@article{gilllevit1995,
  author  = {Gill, Richard D. and Levit, Boris Y.},
  title   = {Applications of the van Trees Inequality: A {B}ayesian {C}ram{\'e}r--{R}ao Bound},
  journal = {Bernoulli},
  volume  = {1},
  number  = {1/2},
  pages   = {59--79},
  year    = {1995}
}

@article{zamir1998fisher,
  author  = {Zamir, Ram},
  title   = {A Proof of the {F}isher Information Inequality via a Data Processing Argument},
  journal = {IEEE Transactions on Information Theory},
  volume  = {44},
  number  = {3},
  pages   = {1246--1250},
  year    = {1998}
}

@article{bottou2018optimization,
  author  = {Bottou, L{\'e}on and Curtis, Frank E. and Nocedal, Jorge},
  title   = {Optimization Methods for Large-Scale Machine Learning},
  journal = {SIAM Review},
  volume  = {60},
  number  = {2},
  pages   = {223--311},
  year    = {2018}
}

@article{ailon2009fast,
  author  = {Ailon, Nir and Chazelle, Bernard},
  title   = {The Fast {J}ohnson--{L}indenstrauss Transform and Approximate Nearest Neighbors},
  journal = {SIAM Journal on Computing},
  volume  = {39},
  number  = {1},
  pages   = {302--322},
  year    = {2009}
}

@book{dembozeitouni1998,
  author    = {Dembo, Amir and Zeitouni, Ofer},
  title     = {Large Deviations Techniques and Applications},
  edition   = {2nd},
  publisher = {Springer},
  year      = {1998}
}

@article{nair2004stabilizability,
  author  = {Nair, Girish N. and Evans, Robin J.},
  title   = {Stabilizability of Stochastic Linear Systems with Finite Feedback Data Rates},
  journal = {SIAM Journal on Control and Optimization},
  volume  = {43},
  number  = {2},
  pages   = {413--436},
  year    = {2004}
}

@article{tanaka2017sdp,
  author  = {Tanaka, Takashi and Kim, Kwang-Ki K. and Parrilo, Pablo A. and Mitter, Sanjoy K.},
  title   = {Semidefinite Programming Approach to {Gaussian} Sequential Rate-Distortion Trade-offs},
  journal = {IEEE Transactions on Automatic Control},
  volume  = {62},
  number  = {4},
  pages   = {1896--1910},
  year    = {2017}
}

@article{stavrou2020revisited,
  author  = {Stavrou, Photios A. and Tanaka, Takashi and Tatikonda, Sekhar},
  title   = {The Time-Invariant Multidimensional {Gaussian} Sequential Rate-Distortion Problem Revisited},
  journal = {IEEE Transactions on Automatic Control},
  volume  = {65},
  number  = {5},
  pages   = {2245--2249},
  year    = {2020}
}

@article{kostina2019ratecost,
  author  = {Kostina, Victoria and Hassibi, Babak},
  title   = {Rate-Cost Tradeoffs in Control},
  journal = {IEEE Transactions on Automatic Control},
  volume  = {64},
  number  = {11},
  pages   = {4525--4540},
  year    = {2019}
}

@inproceedings{zamir1998nested,
  author    = {Zamir, Ram and Shamai, Shlomo},
  title     = {Nested Linear/Lattice Codes for {Wyner--Ziv} Encoding},
  booktitle = {Proc.\ IEEE Information Theory Workshop (ITW)},
  pages     = {92--93},
  year      = {1998}
}

@book{andersonmoore1979,
  author    = {Anderson, Brian D. O. and Moore, John B.},
  title     = {Optimal Filtering},
  publisher = {Prentice-Hall},
  year      = {1979}
}

@article{mayekartyagi2020limits,
  author  = {Mayekar, Prathamesh and Tyagi, Himanshu},
  title   = {Limits on Gradient Compression for Stochastic Optimization},
  journal = {arXiv preprint arXiv:2001.09032},
  year    = {2020},
  note    = {Also in Proc.\ IEEE ISIT 2020}
}

@inproceedings{mayekartyagi2020ratq,
  author    = {Mayekar, Prathamesh and Tyagi, Himanshu},
  title     = {{RATQ}: A Universal Fixed-Length Quantizer for Stochastic Optimization},
  booktitle = {Proc.\ Int.\ Conf.\ on Artificial Intelligence and Statistics (AISTATS)},
  year      = {2020}
}

@article{menartnikolov2025gradient,
  author  = {Menart, Michael and Nikolov, Aleksandar},
  title   = {On the Gradient Complexity of Private Optimization with Private Oracles},
  journal = {arXiv preprint arXiv:2511.13999},
  year    = {2025}
}

\appendix

\noindent The appendix is organized as follows. Appendix~\ref{app:proofs} contains the deferred proofs
of the achievability results (Lemma~\ref{lem:compress} and Theorem~\ref{thm:upper}).
Appendix~\ref{app:exp} reports the numerical sanity checks. Appendix~\ref{app:source} relates the
product-form bound to the distributed-estimation literature, and Appendix~\ref{app:partial} gives a
self-contained Gaussian Fisher-trace bound together with an alternative mutual-information route and
why it is harder. Appendix~\ref{app:seqrd} develops the sequential rate--distortion perspective on
correlated and drifting oracles (Corollary~\ref{cor:correlated} and the tracking bounds).

\section{Deferred Proofs}
\label{app:proofs}

\begin{proof}[Proof of Lemma~\ref{lem:compress}]
The public coin draws a uniformly random subset $S\subseteq[d]$ with $|S|=s$ and a random orthogonal $H$
(a random Hadamard rotation); both are shared by encoder and decoder, so neither costs bits. Crucially,
the quantization grid is the \emph{fixed, a priori} interval $[-G,G]$ determined by the known bound $G$,
\emph{not} by the realized $\|g\|$; hence no scale or norm is transmitted and the message consists only
of the $s$ quantized coordinate indices.

\emph{Encoder.} Compute $\tilde g=Hg$; orthogonality gives $\|\tilde g\|_2=\|g\|_2\le G$, so every
$\tilde g_i\in[-G,G]$. For each $i\in S$, transmit an unbiased stochastic quantization $Q(\tilde g_i)$ on
the uniform grid of $2^q$ points in $[-G,G]$ (spacing $\Delta=2G/(2^q-1)$): for adjacent grid points
$a\le\tilde g_i\le b$ write $\tilde g_i=(1-u)a+ub$ and set $Q(\tilde g_i)=b$ w.p.\ $u$, else $a$. Then
$\E_Q[Q(\tilde g_i)]=\tilde g_i$ and $\mathrm{Var}_Q(Q(\tilde g_i))\le\Delta^2/4\le G^2/2^{2q-2}$.

\emph{Decoder.} Set $\hat{\tilde g}_i=(d/s)\,Q(\tilde g_i)$ for $i\in S$ and $0$ otherwise; output
$\hat g:=\Ccal_B(g)=H^\top\hat{\tilde g}$.

\emph{Unbiasedness.} $\Pr[i\in S]=s/d$ gives $\E[(d/s)\mathbf 1\{i\in S\}]=1$; with
$\E_Q[Q(\tilde g_i)]=\tilde g_i$ and $S\perp Q$, $\E[\hat{\tilde g}]=\tilde g$ and $\E[\hat g]=H^\top\tilde g=g$.

\emph{Second moment.} Using $\|\hat g\|^2=\|\hat{\tilde g}\|^2$ and $\Pr[i\in S]=s/d$,
\[
  \E\|\hat g\|^2=\frac{d}{s}\sum_{i=1}^d\E_Q[Q(\tilde g_i)^2]
  \le\frac{d}{s}\sum_{i=1}^d\Big(\tilde g_i^2+\frac{G^2}{2^{2q-2}}\Big)
  =\frac{d}{s}\Big(\|g\|^2+\frac{dG^2}{2^{2q-2}}\Big)
  \le\frac{d}{s}\|g\|^2+\frac{G^2}{s},
\]
the last step by $2^{2q-2}\ge d^2$. Thus $\E\|\hat g\|^2\le\omega_B\|g\|^2+G^2/s$.

\emph{Bits.} The message is exactly $s\,q\le B$ bits ($s$ indices into a grid of $2^q$ points); no scale
is sent. If $B\ge qd$ then $s=d$, $\omega_B=1$; otherwise $s=\lfloor B/q\rfloor\ge B/(2q)$, so
$\omega_B=d/s\le 2qd/B$. This gives the stated $\omega_B$.
\end{proof}

\begin{proof}[Proof of Theorem~\ref{thm:upper}]
Unbiasedness is the tower property:
$\E[\hat g_t\mid x_t]=\E[\E[\Ccal_B(g_t)\mid g_t]\mid x_t]=\E[g_t\mid x_t]=\nabla f(x_t)$ (the compressor
applies since $\|g_t\|\le G$). For the second moment, Lemma~\ref{lem:compress} gives
$\E[\|\hat g_t\|^2\mid g_t]\le\omega_B\|g_t\|^2+G^2/s$; taking $\E[\cdot\mid x_t]$ and using
$\E[\|g_t\|^2\mid x_t]=\|\nabla f(x_t)\|^2+\E[\|\xi_t\|^2\mid x_t]\le\|\nabla f(x_t)\|^2+\sigma^2 d$,
\[
  \E[\|\hat g_t\|^2\mid x_t]\;\le\;\omega_B\|\nabla f(x_t)\|^2+\underbrace{\omega_B\sigma^2 d+G^2/s}_{=:M}.
\]
This is Assumption~4.3 of \citet{bottou2018optimization} with $M_G=\omega_B$, $\mu_G=1$. Their
Theorem~4.7 (diminishing steps with $\gamma$ chosen so $\eta_0 LM_G\le1$, i.e.\ $\gamma=\Theta(\omega_B L/\mu)$)
gives the displayed $\nu/(\gamma+T)$ bound. Solving $\nu/(\gamma+T)\le\eps^2$ gives
$T=O(\nu/\eps^2)=O\big(\tfrac{LM}{\mu^2\eps^2}+\tfrac{\gamma(f(x_1)-f^\star)}{\eps^2}\big)$. For the
quadratic family with the assumption $G^2=O(\sigma^2 d)$ we have $G^2/s\le\omega_B\sigma^2 d$ (since
$G^2/s=\omega_B G^2/d=O(\omega_B\sigma^2)$), so $M=O(\omega_B\sigma^2 d)$; with $x_1=0$ the initial term
is $O(\gamma\eps^2/\eps^2)=O(\omega_B)$, dominated by $\omega_B\sigma^2 d/\eps^2$ because
$\sigma^2 d/\eps^2\ge1$. Substituting $\omega_B=O(\max\{1,d\log d/B\})$ gives the complexity.
\end{proof}

\section{Numerical Sanity Checks}
\label{app:exp}

The lower bound is a theorem; the experiments below neither can nor attempt to \emph{validate} it.
They are \emph{sanity checks} in the idealized synthetic Gaussian quadratic model---the setting in which
the theory is stated---meant only to confirm that nothing in the bookkeeping is off and to make the
oracle gap concrete. Specifically we check that: the achievability rate of Theorem~\ref{thm:upper} is
observed (Appendix~\ref{app:exp:rate}); the bounded-range oracle gap (L3$'$) manifests as a real
dynamic-range trade-off (Appendix~\ref{app:exp:gap}); the product form, not the $\max$ form, matches the
measured cost (Appendix~\ref{app:exp:product}); three mechanistically different unbiased $B$-bit
compressors all stay on the admissible side of the bound (Appendix~\ref{app:exp:envelope}); bits rather
than dimension cap the transcript Fisher trace (Appendix~\ref{app:exp:fisher}); and the
optimization-to-estimation reduction holds numerically (Appendix~\ref{app:exp:reduction}). None of these
is evidence about non-quadratic or transformer training, about correlated-gradient regimes, or about the
unbounded Gaussian oracle on the achievability side---for that side we impose bounded dynamic range by
clipping before compression, so the algorithm operates on the bounded-range variant of the oracle.

\subsection{The product-form rate is observed}
\label{app:exp:rate}

Setup: $f_\theta(x)=\tfrac12\|x-\theta\|^2$ with $d=128$, $\sigma=1$, target gap $\eps^2=0.1$;
SGD with diminishing step $\eta_t=\beta/(\gamma+t)$, $\beta=2$, $\gamma=\omega_B$; the compressor of
Lemma~\ref{lem:compress} with $q=\lceil\log_2 d\rceil+1=8$ bits per coordinate,
$s=\lfloor B/q\rfloor$ selected coordinates, fixed grid on $[-G,G]$, $G=4\sigma\sqrt d$, and the
bounded-gradient assumption enforced by clipping $g$ to norm $G$ before compression. We sweep
$B\in\{16,32,64,128,256,512,1024\}$, all satisfying the Lemma~\ref{lem:compress} condition
$B\ge 2q=16$, and run $5$ seeds, recording the first iterate that achieves the target gap.

\begin{figure}[h]
  \centering
  \includegraphics[width=0.6\textwidth]{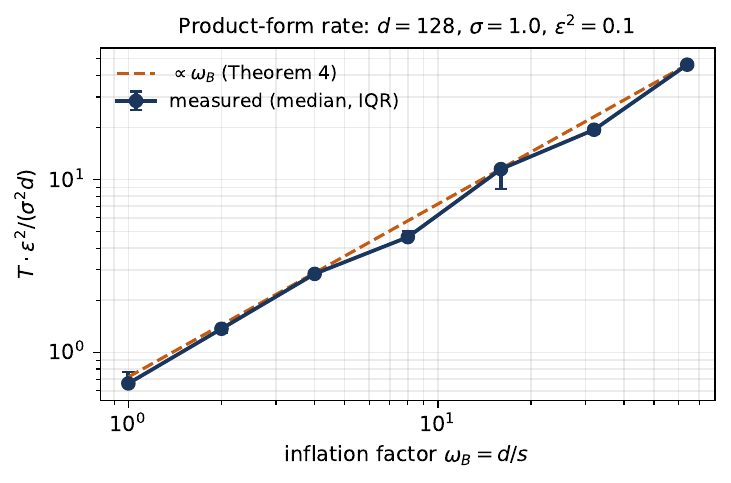}
  \caption{Normalised rounds-to-target track the inflation factor $\omega_B=d/s$ linearly (orange
  dashed line, slope $1$ in log--log), as predicted by Theorem~\ref{thm:upper}'s
  $\widetilde O(\omega_B\sigma^2 d/\eps^2)$. Error bars are interquartile ranges over $5$ seeds.}
  \label{fig:rate}
\end{figure}

Figure~\ref{fig:rate} shows the measured points track the predicted slope-$1$ line cleanly. This is a
consistency check in the model's native setting, not a claim about non-quadratic or
correlated-gradient regimes.

\subsection{The bounded-range oracle gap, illustrated}
\label{app:exp:gap}

We vary $G$ at fixed $B=16$ (so $s=\lfloor B/q\rfloor=2$ and $\omega_B=d/s=64$), $d=128$, target
$\eps^2=0.18$, budget $T=30000$. Two effects compete. For small $G$ almost every Gaussian gradient is
clipped before compression: the compressor of Lemma~\ref{lem:compress} remains unbiased \emph{for the
clipped vector}, but the clipped vector is a biased proxy for the original Gaussian gradient $g_t$, so
the iterate inherits a clipping-induced bias. For large $G$ the fixed-grid spacing $2G/(2^q-1)$ and the
additive variance term $G^2/s$ in $M$ grow, and the algorithm fails to reach the target within $T$.

\begin{figure}[h]
  \centering
  \includegraphics[width=0.6\textwidth]{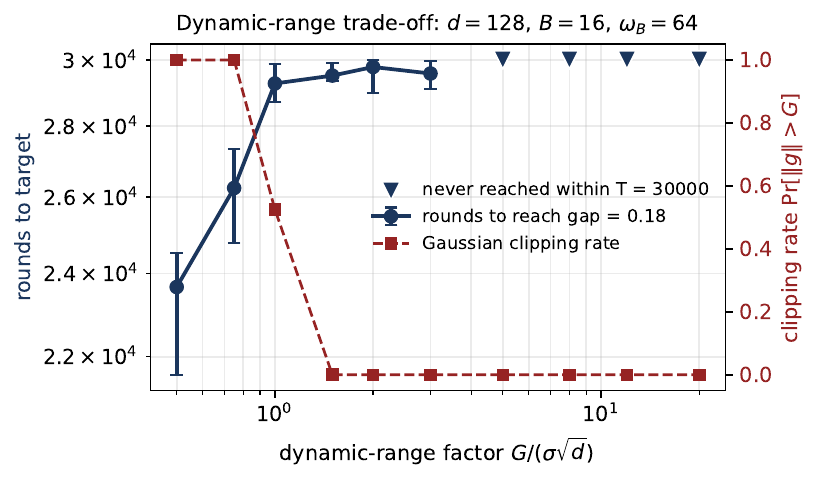}
  \caption{The dynamic-range $G$ is a hidden hyperparameter of the bounded-range oracle. Small
  $G$ (left): the Gaussian clipping rate (red dashed) is $1$, so the clipping operator (not the
  compressor) makes the transmitted gradient a biased proxy for $g_t$. Large $G$ (right, triangles):
  the algorithm never reaches the target gap within $T=30000$. The viable band is narrow,
  $G\approx(1\text{--}3)\,\sigma\sqrt d$. Median over $6$ seeds; bars are IQR.}
  \label{fig:gap}
\end{figure}

The viable $G$-band is narrow---large enough to avoid clipping (so the unbiased bounded-range oracle
applies), small enough to keep the grid fine. Outside the band the algorithm does not silently degrade:
at large $G$ it fails outright. Closing this gap is taken up in Section~\ref{sec:discussion:open}.

\subsection{The product form is not the max form}
\label{app:exp:product}

Section~\ref{sec:product} argues that the elementary max-form bound
$T=\Omega(\max\{\sigma^2 d/\eps^2,\,d/B\})$ (Corollary~\ref{cor:combined}) misses the multiplicative
factor $\max\{1,d/B\}$ that Theorem~\ref{thm:product} supplies. We make this concrete. With $d=128$,
$\sigma=1$, target $\eps^2=0.1$, we sweep $B\in\{16,\dots,1024\}$ and record the median rounds-to-target
over $8$ seeds (the achievability scheme of \S\ref{sec:upper}). Figure~\ref{fig:exp_product} plots the
normalised cost $T\eps^2/(\sigma^2 d)$ against $B$, together with the two competing \emph{scalings}
(constants and the $\log d$ factor dropped, since the point is the $B$-dependence).

\begin{figure}[h]
  \centering
  \includegraphics[width=0.62\textwidth]{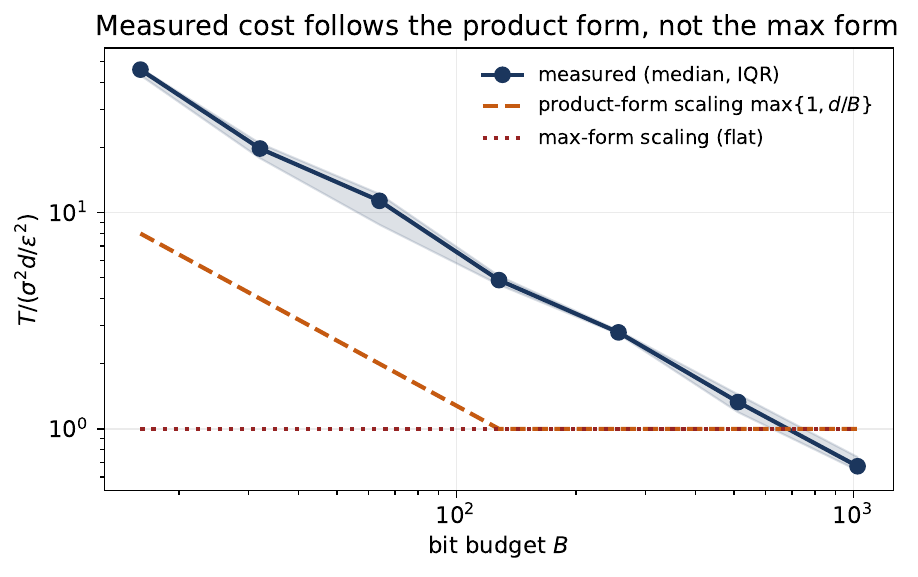}
  \caption{Measured cost tracks the product-form scaling $\max\{1,d/B\}$ (orange dashed) and peels away
  from the flat max-form scaling (red dotted). The vertical offset of the measured points above the
  product scaling is the $\log d$ factor and constants of Theorem~\ref{thm:upper}. Median over $8$ seeds;
  shaded band is IQR.}
  \label{fig:exp_product}
\end{figure}

The max-form scaling is essentially flat (the statistical floor $\sigma^2 d/\eps^2$ dominates the term
$d/B$ for these $B$), so it underestimates the measured cost by a factor that grows as $B$ shrinks:
empirically $45.8\times$ at $B=16$, falling to $\approx 1\times$ once $B\gtrsim d$. The measured points
instead rise parallel to the product-form scaling. This is exactly the deficiency that
\S\ref{sec:product} repairs: a coarse message inflates the effective variance, producing the
multiplicative $\max\{1,d/B\}$ that the max-form cannot see.

\subsection{The lower bound binds every reasonable algorithm}
\label{app:exp:envelope}

Theorem~\ref{thm:product} asserts that \emph{no} $B$-bit algorithm reaches accuracy $\eps$ in fewer than
$\Omega((\sigma^2 d/\eps^2)\max\{1,d/B\})$ rounds. We stress-test this against three mechanistically
different \emph{unbiased} $B$-bit compressors: (i) the rand-$k$ + stochastic-quantization scheme of
Lemma~\ref{lem:compress}; (ii) importance sampling (coordinates drawn with probability $\propto|g_i|$,
Horvitz--Thompson rescaled to stay unbiased) + stochastic quantization; and (iii) a random Hadamard
rotation followed by rand-$k$ + stochastic quantization (the FP4-style coordinate-balancing device).
Each is run to the target $\eps^2=0.1$ over $6$ seeds for $B\in\{32,\dots,512\}$.

\begin{figure}[h]
  \centering
  \includegraphics[width=0.62\textwidth]{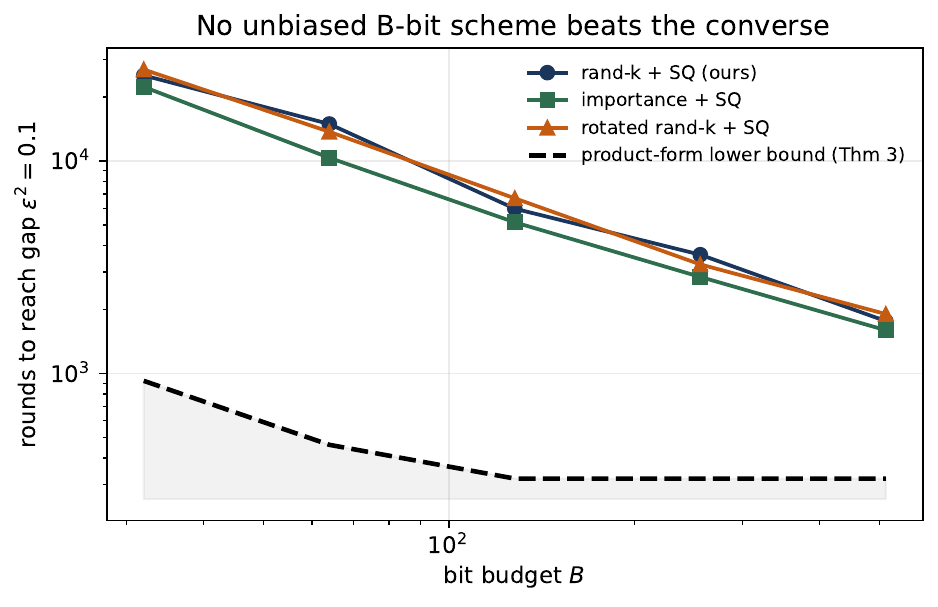}
  \caption{Three unbiased $B$-bit schemes (markers) all lie above the product-form lower bound of
  Theorem~\ref{thm:product} (dashed; shaded region is forbidden) at every $B$. The kink in the bound is
  the $\min\{2\ln2\,B,\,d\}$ crossover. Median over $6$ seeds.}
  \label{fig:exp_envelope}
\end{figure}

All three schemes sit above the lower-bound floor at every $B$ (Figure~\ref{fig:exp_envelope}); none
pierces the converse, and the best tracks it up to the logarithmic factor of
Theorem~\ref{thm:upper}. This is consistent with the bound being a property of the $B$-bit oracle rather
than of any particular compressor.

\subsection{Bits, not dimension, limit Fisher information}
\label{app:exp:fisher}

The engine of Theorem~\ref{thm:product} is Lemma~\ref{lem:fishertranscript}: a $B$-bit message of a
Gaussian observation carries Fisher trace at most $\min\{2\ln2\cdot B,\,d\}/\sigma^2$. Both branches are
verified exactly (no Monte-Carlo) on the scalar/independent Gaussian-location model in
Figure~\ref{fig:exp_fisher}: $B$ independent one-bit signs give Fisher trace growing linearly in $B$
with slope $2/\pi$, safely below the bound's slope $2\ln2$; a single coordinate quantized at $b$ bits
saturates at $1/\sigma^2$ regardless of $b$. The $\min\{2\ln2\,B,\,d\}$ is precisely the crossover
between these two regimes---the mechanism by which dimension $d$ is replaced by bits $B$.

\begin{figure}[h]
  \centering
  \includegraphics[width=0.6\textwidth]{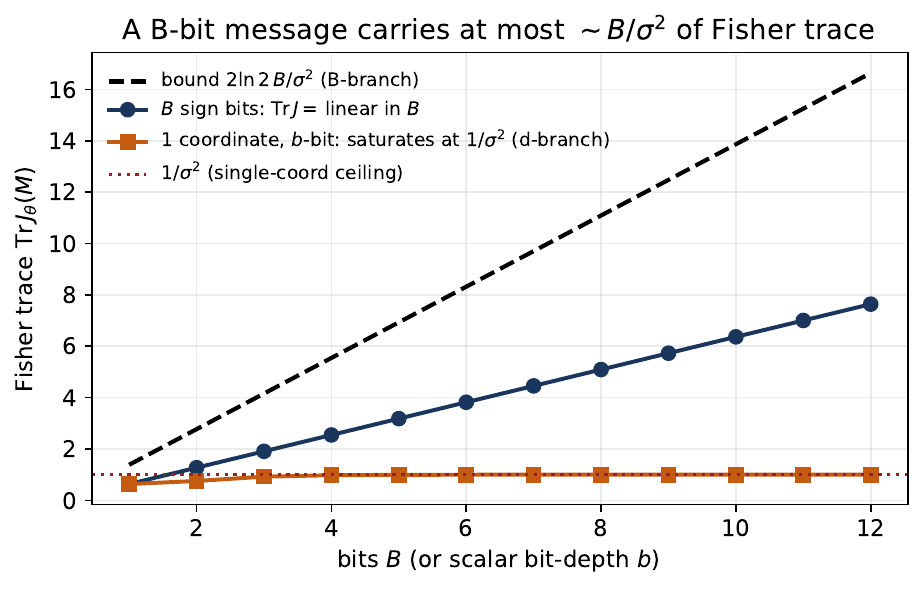}
  \caption{Exact Fisher trace of quantized Gaussian-location channels. The $B$-branch (independent
  sign bits) grows linearly under the $2\ln2\,B/\sigma^2$ bound; the $d$-branch (one coordinate, $b$
  bits) saturates at the single-coordinate ceiling $1/\sigma^2$.}
  \label{fig:exp_fisher}
\end{figure}

\subsection{The reduction is exact}
\label{app:exp:reduction}

Finally we confirm Lemma~\ref{lem:reduction} numerically. The optimization view (SGD on
$f_\theta(x)=\tfrac12\|x-\theta\|^2$ with a compressed gradient) and the compressed-Gaussian-mean-estimation
view are the same stochastic recursion; run on identical noise they coincide to machine precision
(Figure~\ref{fig:exp_reduction}, left). The reduction's premise---that the query carries no
information---is shown on the right: the encoder's observation $Y_t=x_t-g_t$ has the same
$\Ncal(\theta,\sigma^2 I)$ law whether the query $x_t$ is fixed at the origin, random, or drifting, so
the empirical $\|\bar Y-\theta\|$ is essentially identical ($0.051$, $0.055$, $0.057$) across the three
strategies.

\begin{figure}[h]
  \centering
  \includegraphics[width=0.92\textwidth]{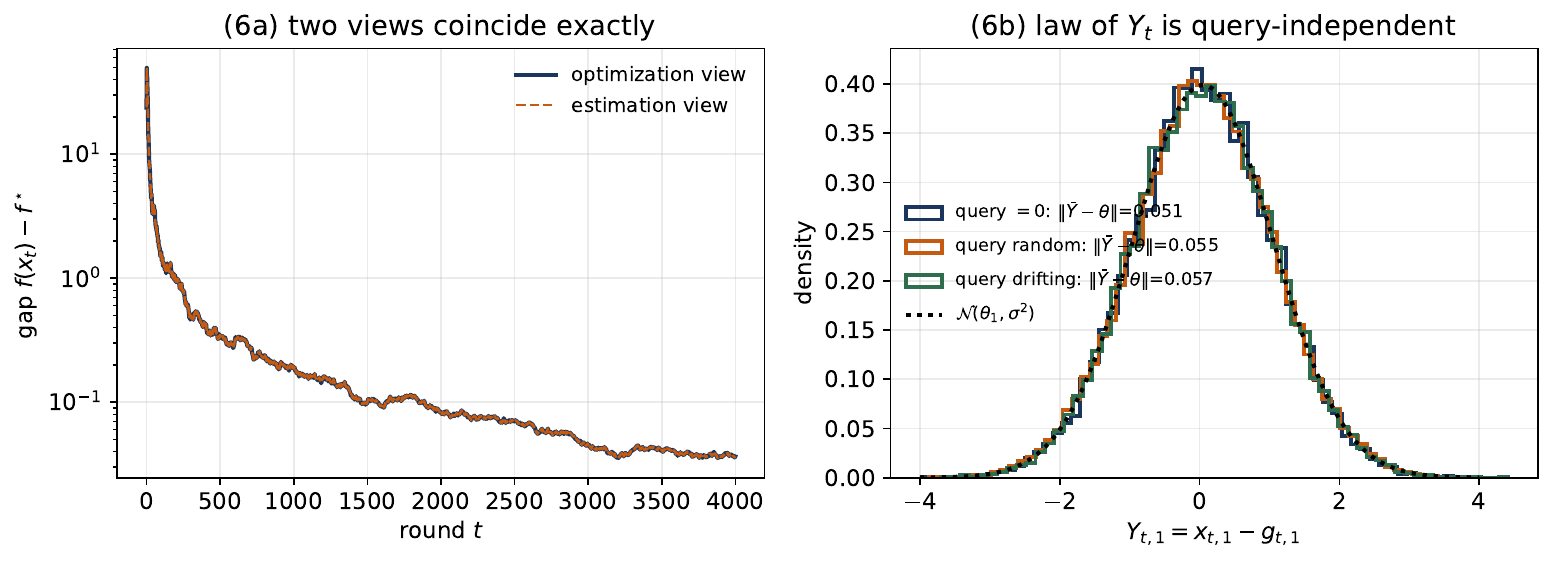}
  \caption{Left: the optimization-view and estimation-view trajectories coincide exactly (max difference
  $0$). Right: the law of the encoder's observation $Y_t=x_t-g_t$ is $\Ncal(\theta,\sigma^2 I)$
  independent of the query strategy, confirming the query is informationally inert.}
  \label{fig:exp_reduction}
\end{figure}

\FloatBarrier

\section{Relation to the Fisher-Information Route and Prior Distributed-Estimation Bounds}
\label{app:source}

The product-form bound (Theorem~\ref{thm:product}) is obtained through the Fisher-information route of
\citet{barnes2021fisher}. We record here how this relates to the distributed-estimation literature and
why we use the Fisher route rather than a direct mutual-information/SDPI argument.

\begin{itemize}[leftmargin=1.4em,itemsep=2pt]
  \item \citet{braverman2016ddpi} prove a \emph{distributed data-processing inequality} for decentralized
  estimation under the Gaussian location model, yielding a communication-vs-error tradeoff; their bound
  is stated through an SDPI constant and requires a bounded-likelihood-ratio assumption, which for the
  Gaussian model forces a truncation argument.
  \item \citet{han2018geometric} give a geometric / Fisher-information bound showing the communication
  constraint reduces the effective sample size in the small-bit regime, for general (including dense)
  families in the sequential/blackboard model; \citet{barneshanozgur2020} recover and extend these
  results explicitly through Fisher information.
  \item \citet{barnes2021fisher} prove the cleanest form for our purposes: if the score is sub-Gaussian
  (which the Gaussian location score is, exactly), the transcript Fisher trace is at most
  $2N^2$ times the transcript mutual information (their Theorem~1 and interactive Corollary~1), with the
  optimal preconstant $2/\sigma^2$ for the Gaussian model (their Corollary~2), and combine it with the
  multivariate van Trees inequality \citep{gilllevit1995} to obtain minimax lower bounds (their
  Corollary~3). Crucially this applies to the Gaussian model \emph{directly}, with no truncation.
\end{itemize}

\noindent\textbf{Why the Fisher route and not a mutual-information/SDPI argument.} A natural alternative route
bounds the product form through a transcript mutual-information contraction
$I(V;M_{1:T}\mid U)\lesssim TB\delta^2/\sigma^2$ for the hypercube prior. As Appendix~\ref{app:partial}
documents, the $\chi^2$ version of this contraction is controllable (it equals a maximal correlation and
is $O(\mathrm{SNR})$), but the mutual-information (KL) version is provably no smaller and is not bounded from above by the same
argument; moreover the naive Fisher-to-mutual-information conversion
$I\le\tfrac12\Tr(\Sigma_\pi I_\Pi)$ \emph{fails} at finite SNR because it omits the prior-information term
that van Trees supplies. The van Trees inequality resolves exactly this: its denominator carries the
prior Fisher information $\pi^2 d/\delta^2$ alongside $\Tr I_\Pi(\theta)$, which is what the naive
conversion dropped. This is why Theorem~\ref{thm:product} is unconditional while the
mutual-information route of Appendix~\ref{app:partial} remains partial.

\section{The Gaussian Fisher-Trace Bound, and an Alternative Mutual-Information Route}
\label{app:partial}

This appendix has two purposes. First (\S\ref{app:d:6}) it gives a short, self-contained proof of the Gaussian
single-round case of the Fisher-trace bound used in Lemma~\ref{lem:fishertranscript}, via a
Donsker--Varadhan centroid inequality. Second, it records an alternative route to the product-form bound
through a transcript \emph{mutual-information} contraction---the natural SDPI-style argument---and
documents precisely why that route is harder than the Fisher route of the main text, so that a reader
who prefers the mutual-information formulation can see exactly where it stalls. Throughout,
$\mathrm{SNR}:=\delta^2/\sigma^2$ and $\mathrm{SNR}\le1$. The target inequality of this alternative route
is the transcript contraction
\begin{equation}
  I(V;M_{1:T}\mid U)\;\le\;C_\star(d)\,\frac{TB\,\delta^2}{\sigma^2}\qquad(\text{bits}),\qquad C_\star(d)=\widetilde O(1),
  \tag{$\star$}\label{eq:sdpi}
\end{equation}
which would give the product form by the Fano argument of \S\ref{sec:product}; the main text instead
avoids \eqref{eq:sdpi} via van Trees.

\subsection{Single-coordinate facts (proved)}\label{app:d:1}

\begin{lemma}[Single-coordinate $\chi^2$ and mutual information]
\label{lem:single}
Let $V_1\sim\Unif\{\pm1\}$ and $Y_1=\delta V_1+\Ncal(0,\sigma^2)$. Then
\[
  \chi^2\!\big(\Ncal(\delta,\sigma^2)\,\big\|\,\Ncal(-\delta,\sigma^2)\big)=e^{4\delta^2/\sigma^2}-1
  =4\,\mathrm{SNR}+O(\mathrm{SNR}^2),
\]
and
\[
  I(V_1;Y_1)\le\tfrac12\log_2\!\big(1+\mathrm{SNR}\big)\le\frac{\mathrm{SNR}}{2\ln2}\quad(\text{bits}).
\]
\end{lemma}

\begin{proof}
For equal-variance Gaussians, $\chi^2(\Ncal(\mu_1,\sigma^2)\|\Ncal(\mu_0,\sigma^2))=e^{(\mu_1-\mu_0)^2/\sigma^2}-1$;
here $\mu_1-\mu_0=2\delta$. The mutual-information bound is Proposition~\ref{prop:elem}: matching the
variance $\mathrm{Var}(Y_1)=\sigma^2+\delta^2$ in the maximum-entropy inequality gives
$I(V_1;Y_1)\le\tfrac12\log_2(1+\mathrm{SNR})$, and $\log_2(1+x)\le x/\ln2$.
\end{proof}

\subsection{The product bound is tight (proved)}\label{app:d:2}

\begin{proposition}[Achievability and tightness of the SNR-linear rate]
\label{prop:tight}
For the $d$-coordinate instance with $1\le B\le d$ (so the coordinates $Y_1,\dots,Y_B$ exist), the
$B$-bit message $M=(\sign Y_1,\dots,\sign Y_B)$ satisfies
\[
  I(V;M)=B\,\big(1-\hb(\Phi(-\delta/\sigma))\big)\;=\;c(\mathrm{SNR})\,\cdot\,\mathrm{SNR}\,\cdot\,B,
  \qquad c(\mathrm{SNR})\xrightarrow[\mathrm{SNR}\to0]{}\frac{1}{\pi\ln2}\approx0.459,
\]
where $\Phi$ is the standard normal CDF. Hence, in the single-round, low-SNR, $1\le B\le d$ regime, the
right-hand side of the target inequality~\eqref{eq:sdpi} is attained up to a constant: \emph{if} the converse holds,
then its constant is $\Theta(1)$; equivalently, the target upper bound in~\eqref{eq:sdpi},
if true, is sharp in its $B$ and $\mathrm{SNR}$ dependence.
\end{proposition}

\begin{proof}
The coordinates are independent, so $I(V;M)=\sum_{j\le B}I(V_j;\sign Y_j)=B\,(1-\hb(p))$ with
$p=\Prob[\sign Y_1\neq V_1]=\Phi(-\delta/\sigma)$. As $\mathrm{SNR}\to0$, $p=\tfrac12-\tfrac{\delta/\sigma}{\sqrt{2\pi}}+O(\mathrm{SNR}^{3/2})$,
and $1-\hb(\tfrac12-\epsilon)=\tfrac{2}{\ln2}\epsilon^2+O(\epsilon^4)$ with $\epsilon=\tfrac{\delta/\sigma}{\sqrt{2\pi}}$,
giving $1-\hb(p)=\tfrac{1}{\pi\ln2}\,\mathrm{SNR}+O(\mathrm{SNR}^2)$.
\end{proof}

\noindent Thus the only question is the \emph{converse}: that \emph{no} $B$-bit function of
$Y_{1:d}$ (across $T$ adaptive rounds) extracts more than $O(\mathrm{SNR}\cdot B)$ per round.

\subsection{Why elementary arguments stop at the $\max$-form}\label{app:d:3}

Two ceilings are elementary: $I(V;M_{1:T}\mid U)\le TB$ (Theorem~\ref{thm:comm}) and
$I(V;M_{1:T}\mid U)\le\tfrac{Td\,\mathrm{SNR}}{2\ln2}$ (Proposition~\ref{prop:elem}). Their minimum yields only
Corollary~\ref{cor:combined}. In the regime $1\le B\le d$ the binding ceiling is the communication one,
$\min\{TB,\,Td\,\mathrm{SNR}\}=TB$ (when $B\le d\,\mathrm{SNR}$), whereas the product target is
$TB\,\mathrm{SNR}$, smaller by a factor $\mathrm{SNR}^{-1}$. Elementary arguments therefore overestimate
the surviving information by a factor $\Theta(\mathrm{SNR}^{-1})$: each transmitted bit is treated as
fully informative about $V$, but a bit computed from the noise-diluted $Y$ can carry only an
$\mathrm{SNR}$-fraction of signal information. Closing this gap is precisely the content of the
distributed data-processing inequality and is not elementary.

\subsection{A single-round contraction: what the maximal-correlation route does and does not give}\label{app:d:4}

Reversing the Markov chain turns the ``downstream SDPI'' obstacle into a question about the backward
channel $P_{V\mid Y}$, which we can analyze partially. The outcome is an honest one: the
\emph{$\chi^2$} contraction is controlled by maximal correlation and is $O(\mathrm{SNR})$, but the
\emph{KL/mutual-information} contraction we actually need is provably no smaller and we do not bound it
from above. We state the rigorous facts and the precise gap.

\begin{lemma}[Rigorous ingredients, single round]
\label{lem:singleround}
Let $V\sim\Unif(\{\pm1\}^d)$, $Y=\delta V+\Ncal(0,\sigma^2 I_d)$, $U\perp(V,Y)$, and $M=Q(Y,U)\in\{0,1\}^B$.
Write $\rho_m(V,Y)$ for the Hirschfeld--Gebelein--R\'enyi maximal correlation, and let
$\eta_{\mathrm{KL}}(P_Y,P_{V\mid Y})$ and $\eta_{\chi^2}(P_Y,P_{V\mid Y})$ denote the input-distribution-dependent
KL- and $\chi^2$-contraction coefficients of the backward channel $P_{V\mid Y}$ taken at the specific
input marginal $P_Y$ (i.e.\ $\eta_f=\sup_{Q_Y\neq P_Y}D_f(Q_YP_{V\mid Y}\Vert P_YP_{V\mid Y})/D_f(Q_Y\Vert P_Y)$).
Then:
\begin{enumerate}[leftmargin=1.6em,itemsep=2pt]
  \item \emph{(Reversal and SDPI.)} $V\to Y\to M$ is equivalent to $M\to Y\to V$, and the strong
  data-processing inequality \citep{polyanskiy2017sdpi} for the backward channel $P_{V\mid Y}$ gives
  $I(V;M\mid U)\le\eta_{\mathrm{KL}}(P_Y,P_{V\mid Y})\,I(Y;M\mid U)\le\eta_{\mathrm{KL}}(P_Y,P_{V\mid Y})\,B$,
  using $I(Y;M\mid U)\le H(M\mid U)\le B$ and $U\perp(V,Y)$.
  \item \emph{($\chi^2$ contraction $=$ maximal correlation.)} $\eta_{\chi^2}(P_Y,P_{V\mid Y})=\rho_m^2(V,Y)$,
  and maximal correlation tensorizes over the (identical, independent) coordinates,
  $\rho_m^2(V,Y)=\max_j\rho_m^2(V_j,Y_j)=\rho_m^2(V_1,Y_1)$.
  \item \emph{(Scalar bound.)} $\rho_m^2(V_1,Y_1)=\E[\tanh^2(\delta Y_1/\sigma^2)]\le\frac{\delta^2}{\sigma^4}\E[Y_1^2]
  =\mathrm{SNR}(1+\mathrm{SNR})\le2\,\mathrm{SNR}$ for $\mathrm{SNR}\le1$, using $\tanh^2(x)\le x^2$ and the
  posterior mean $\E[V_1\mid Y_1]=\tanh(\delta Y_1/\sigma^2)$.
\end{enumerate}
\end{lemma}

The following three statements are standard or directly verified. The difficulty is that they do
\emph{not} compose into the bound we want. The SDPI in item~1 carries the \emph{KL} contraction
coefficient $\eta_{\mathrm{KL}}(P_Y,P_{V\mid Y})$, whereas items~2--3 bound the \emph{$\chi^2$} coefficient
$\eta_{\chi^2}=\rho_m^2$. Under the distribution-dependent SDPI convention of
\citet{polyanskiy2017sdpi}, the local (small-perturbation) KL contraction at $P_Y$ equals the $\chi^2$
contraction, so the global coefficient satisfies $\eta_{\mathrm{KL}}(P_Y,P_{V\mid Y})\ge\eta_{\chi^2}(P_Y,P_{V\mid Y})=\rho_m^2$;
thus $\rho_m^2$ is a \emph{lower} bound on the relevant coefficient, not an upper bound. Concretely, the
operational identity $\eta_{\mathrm{KL}}(P_Y,P_{V\mid Y})=\sup_{M:\,V-Y-M} I(V;M)/I(Y;M)$ is realized by
skewed low-entropy (rare-event) messages, for which the ratio $I(V;M)/I(Y;M)$ strictly exceeds
$\rho_m^2$; thus the clean inequality $I(V;M)\le\rho_m^2\,H(M)$ is \emph{false}. What survives is the
implication chain $I(V;M\mid U)\le\eta_{\mathrm{KL}}(P_Y,P_{V\mid Y})\,B$, which yields the single-round
product bound \emph{iff} $\eta_{\mathrm{KL}}(P_Y,P_{V\mid Y})=O(\mathrm{SNR})$.

\begin{conjecture}[Single-round KL contraction is $O(\mathrm{SNR})$]
\label{conj:singleround}
For the binary hypercube Gaussian-location channel with $\mathrm{SNR}\le1$,
$\eta_{\mathrm{KL}}(P_Y,P_{V\mid Y})\le C\,\mathrm{SNR}$ for an absolute constant $C$. Under this conjecture,
$I(V;M\mid U)\le C\,\mathrm{SNR}\cdot B$, i.e.\ the single-round ($T=1$) case of
the target inequality~\eqref{eq:sdpi} holds with $C_\star(d)=C$ (no $\mathrm{polylog}$ at $T=1$).
\end{conjecture}

\noindent We do not prove Conjecture~\ref{conj:singleround}. We know $\eta_{\mathrm{KL}}\ge\rho_m^2=\Theta(\mathrm{SNR})$
(item~2--3), so if true the rate is the right one; bounding $\eta_{\mathrm{KL}}$ \emph{from above} by
$O(\mathrm{SNR})$ is an instance of the same Gaussian-channel strong-data-processing estimate that the
distributed-DPI literature establishes by more involved means. We flag this honestly: the
maximal-correlation argument settles the $\chi^2$ contraction but not the mutual-information contraction,
and therefore does \emph{not} by itself make even the single-round bound self-contained.

\subsection{The multi-round gap}\label{app:d:5}

Even granting Conjecture~\ref{conj:singleround}, the multi-round bound does not follow by naive
telescoping. By the chain rule $I(V;M_{1:T}\mid U)=\sum_t I(V;M_t\mid M_{1:t-1},U)$, and conditioned on a
prefix the round-$t$ term involves $V$ drawn from the posterior $\mu_m:=P(V\mid M_{1:t-1}=m,U)$ rather
than the uniform prior. The relevant coefficient is then that of the \emph{posterior} channel, and the
maximal correlation under $\mu_m$ need not stay $O(\mathrm{SNR})$: if a prefix could force the coordinates
of $V$ to be strongly correlated, a fresh Gaussian look would behave like many looks at one bit and the
per-round contraction would degrade. Controlling this is exactly the inductive information-cost step of
the distributed-DPI argument, which we do not carry out.

\begin{remark}[Numerical observations, not used in any proof]
\label{rem:numerics}
The following exploratory computations are not invoked anywhere above. For \emph{product} posteriors
(independent, possibly tilted coordinates) the maximal correlation appears to \emph{decrease} as the
posterior concentrates (e.g.\ at $\mathrm{SNR}=0.1$ it falls from $\approx0.09$ toward $0.01$ as a
coordinate's posterior mass moves from $0.5$ to $0.97$). For \emph{strongly correlated} posteriors it
appears to grow (e.g.\ at $d=8$, $\mathrm{SNR}=0.1$ it rises from $\approx0.09$ for independent
coordinates to $\approx0.44$ when all coordinates are forced equal, exceeding $2\,\mathrm{SNR}=0.2$).
These observations suggest that multi-round safety hinges on whether a bounded-bit prefix can induce
strongly correlated posteriors; we make no formal claim and use none of this in any proof.
\end{remark}

\noindent Hence multi-round safety reduces to a single inductive statement: \emph{a transcript of
$(t-1)B$ bits cannot drive the posterior $\mu_m$ to a maximal correlation with a fresh Gaussian look
exceeding $O(\mathrm{SNR})$.} This is precisely the content the distributed data-processing inequality of
\citet{braverman2016ddpi} supplies through its information-cost induction, and is the one ingredient we
do not establish. We note that an earlier draft of this analysis incorrectly claimed the per-round
contraction is automatically uniform over prefixes ``because the SDPI coefficient is a channel
property''; the maximal-correlation computation above shows that claim is false---the coefficient
depends on the (posterior) input distribution---and isolates the genuine difficulty.

\subsection{A rigorous compression bound via Gaussian centroids}\label{app:d:6}

The maximal-correlation route of \S\ref{app:d:4} stalls because it targets the contraction coefficient directly.
A different route proves cleanly the statement that ``$B$ bits carry only $B$ dimensions' worth of
Gaussian score energy,'' which is the mechanism by which $d$ is replaced by $B$. This does not by itself
close the target inequality~\eqref{eq:sdpi}, but it isolates the remaining gap more sharply.

\begin{lemma}[Gaussian centroid bound]
\label{lem:centroid}
Let $Z\sim\Ncal(0,I_d)$ and let $M$ be any random variable with $Z\to M$ (a channel output of $Z$).
Then $\E\big\|\E[Z\mid M]\big\|^2\le 2\,I(Z;M)\le 2\,H(M)$ (nats).
\end{lemma}

\begin{proof}
For any probability measure $Q\ll\gamma:=\Ncal(0,I_d)$, the Donsker--Varadhan variational formula
\citep[Lemma~6.2.13]{dembozeitouni1998} with the
linear test function $f(z)=\langle a,z\rangle$ (for which $\log\E_\gamma e^{f}=\tfrac12\|a\|^2$) gives
$D(Q\Vert\gamma)\ge\langle a,\E_Q Z\rangle-\tfrac12\|a\|^2$; maximizing over $a$ at $a=\E_Q Z$ yields
$D(Q\Vert\gamma)\ge\tfrac12\|\E_Q Z\|^2$. Apply this with $Q=P_{Z\mid M=m}$ and average:
$I(Z;M)=\E_m D(P_{Z\mid M=m}\Vert\gamma)\ge\tfrac12\E_m\|\E[Z\mid M=m]\|^2$. Finally $I(Z;M)\le H(M)$.
\end{proof}

\begin{lemma}[Bits limit Fisher trace]
\label{lem:fishertrace}
Let $Y=\theta+\sigma Z$ with $Z\sim\Ncal(0,I_d)$, and let $M=Q(Y,U)\in\{0,1\}^B$ with $U\perp(Z)$. Let
$J_\theta(M\mid U)$ be the Fisher information of the message about $\theta$. Then at every $\theta$,
\[
  \Tr J_\theta(M\mid U)\;=\;\frac{1}{\sigma^2}\,\E\big\|\E[Z\mid M,U]\big\|^2\;\le\;\frac{2}{\sigma^2}\,I(Z;M\mid U)\;\le\;\frac{2\ln2}{\sigma^2}\,B .
\]
For an adaptive $T$-round protocol with fresh noise each round, Fisher information adds across rounds, so
$\Tr J_\theta(M_{1:T}\mid U)\le 2\ln2\,TB/\sigma^2$.
\end{lemma}

\begin{proof}
At fixed $\theta$ and fixed $U=u$, the score for outcome $m$ is
$\nabla_\theta\log P_\theta(M=m)=\sigma^{-2}\E[Y-\theta\mid M=m]=\sigma^{-1}\E[Z\mid M=m]$, so
$\Tr J_\theta(M\mid U{=}u)=\sigma^{-2}\E\|\E[Z\mid M,u]\|^2$; Lemma~\ref{lem:centroid} (with $Z$ a function
of fresh noise at fixed $\theta,u$) bounds this by $\tfrac{2}{\sigma^2}I(Z;M\mid u)\le\tfrac{2\ln2}{\sigma^2}H(M\mid u)\le\tfrac{2\ln2}{\sigma^2}B$;
average over $u$. Across rounds, the transcript score is the sum of conditional per-round scores, which
are martingale increments, so the Fisher informations add and each round contributes at most
$2\ln2\,B/\sigma^2$.
\end{proof}

\noindent Lemma~\ref{lem:fishertrace} is the rigorous form of the compression mechanism: regardless of
dimension, a $B$-bit-per-round transcript carries at most $O(TB/\sigma^2)$ total Fisher trace about the
mean, not $O(Td/\sigma^2)$. \emph{What remains} to obtain the target inequality~\eqref{eq:sdpi} is a conversion from
this Fisher-trace bound to the hypercube-prior mutual information $I(V;M_{1:T}\mid U)$ with $\theta=\delta V$.
The natural Gaussian conversion $I(V;M)\le\tfrac12\Tr(\Sigma_\pi J)=\tfrac{\delta^2}{2}\Tr J$ (prior
covariance $\Sigma_\pi=\delta^2 I_d$) would finish the single-round case, but it is \emph{not} valid at
finite SNR: for the hypercube prior $I(V;M)$ can exceed $\tfrac{\delta^2}{2}\Tr J$ (the bound is the
small-perturbation limit, and the discrete prior contributes positive higher-order terms). The weaker
form $I(V;M)\le c\,\delta^2 H(M)$, and the coordinate-Hellinger sensitivity
$\sum_j\E_{V_{-j}}H^2(\,\cdot\,)\le c\,\tfrac{\delta^2}{\sigma^2}H(M)$, are the natural targets, but a rigorous
proof of either in the regime $\mathrm{SNR}\le1$ (rather than $\mathrm{SNR}\to0$) is the residual step. Thus Lemma~\ref{lem:fishertrace} reduces the otherwise opaque assumption~\eqref{eq:sdpi} to a single,
concrete finite-SNR Fisher-to-Bayes conversion (which the main text bypasses entirely via van Trees).

\subsection{Status summary}\label{app:d:7}

We have reduced the target mutual-information inequality~\eqref{eq:sdpi} as follows. The single-round case is \emph{not} proved: the
maximal-correlation route (Lemma~\ref{lem:singleround}) controls the $\chi^2$ contraction but not the
mutual-information contraction, leaving Conjecture~\ref{conj:singleround} (a single-round
$O(\mathrm{SNR})$ bound on $\eta_{\mathrm{KL}}$) open. The Gaussian-centroid route (\S\ref{app:d:6}) instead proves
rigorously that bits limit Fisher trace (Lemma~\ref{lem:fishertrace}), reducing the assumption to a
finite-SNR Fisher-to-Bayes conversion that we verify numerically but do not prove. The multi-round case
requires, in addition, an inductive control of how a bounded-bit prefix can correlate the posterior
(\S\ref{app:d:5}), which is exactly the distributed-DPI core. What is rigorous: the single-coordinate facts
(Lemma~\ref{lem:single}), the $\chi^2$/maximal-correlation contraction (Lemma~\ref{lem:singleround}), the
Gaussian centroid and Fisher-trace bounds (Lemmas~\ref{lem:centroid}--\ref{lem:fishertrace}), and the
tightness of the target rate (Proposition~\ref{prop:tight}, so the exponents cannot be improved). Closing
the finite-SNR conversion and the multi-round induction would complete this alternative
mutual-information route; the main-text product-form bound (Theorem~\ref{thm:product}) is already
unconditional via van Trees and does not depend on it.

\section{Sequential Rate--Distortion Theory for Drifting and Correlated Oracles}
\label{app:seqrd}

This appendix develops a dynamic counterpart of the static reduction (Lemma~\ref{lem:reduction}) and uses
it to (i) \emph{resolve and correct} the correlated-gradient conjecture of \S\ref{sec:discussion:open},
and (ii) outline how a drifting optimum can be tracked under a bit budget. We present it as a
\emph{perspective}, not as core theorems: the tracking \emph{lower bound}
(Theorem~\ref{thm:seqlb}) is self-contained, whereas the converse and achievability invoke
\emph{remote} (indirect) sequential rate--distortion and entropy-coded quantization, which depart from
the exact fixed-length $B$-bit oracle of Definition~\ref{def:oracle} in ways we flag explicitly. The
information-theoretic machinery is imported: directed information \citep{massey1990directed}, the
data-rate theorem \citep{nair2004stabilizability}, Gaussian sequential rate--distortion (SRD) and its
semidefinite representation \citep{tatikonda2009capacity,tanaka2017sdp,stavrou2020revisited}, and
innovation-based lattice achievability \citep{kostina2019ratecost,zamir1998nested}. The new
contributions are the reduction \emph{from} bit-constrained optimization (Lemma~\ref{lem:seqreduce}) and
the resulting correction (Corollary~\ref{cor:correlated}).

\subsection{A drifting quadratic oracle}\label{app:seqrd:setup}
\begin{definition}[Linear--Gaussian oracle]
\label{def:drift}
The optimum follows linear--Gaussian dynamics
$\theta_{t+1}=A\theta_t+w_t$, $w_t\stackrel{\mathrm{iid}}{\sim}\Ncal(0,W)$, with $A$ a real $d\times d$
matrix (we allow unstable modes; the typical optimization-tracking regime has $\rho(A)<1$) and
$W\succ0$ (the singular case follows by a limiting argument); the round-$t$ loss is
$f_t(x)=\tfrac12\|x-\theta_t\|^2$. The query $x_t$ is
$\mathcal F_{t-1}:=\sigma(M^{t-1},U)$-measurable; the oracle returns
$g_t=(x_t-\theta_t)+\xi_t$ with $\xi_t\stackrel{\mathrm{iid}}{\sim}\Ncal(0,V)$, $V\succ0$. The encoder
forms $Y_t:=x_t-g_t=\theta_t-\xi_t$ and sends $M_t=\mathcal E_t(Y^t,M^{t-1},U)\in\{0,1\}^B$; the
optimizer sets $x_{t+1}=\mathcal D_t(M^t,U)$, so $\hat\theta_t:=x_t$ is $\mathcal F_{t-1}$-measurable
(a one-step predictor). Write $D_t:=\tfrac12\E\|x_t-\theta_t\|^2$ and
$\bar D_\infty:=\limsup_{T}\tfrac1T\sum_{t\le T}D_t$.
\end{definition}

\begin{remark}[One model, two regimes]
\label{rem:augment}
Two regimes fit this template. The \emph{drifting optimum} (Model B; the focus of
\S\ref{app:seqrd:lb1}--\ref{app:seqrd:ach}) is Definition~\ref{def:drift} directly. The \emph{static
optimum with temporally correlated noise} (Model A) is handled in the main text
(Corollary~\ref{cor:correlated}) by whitening; equivalently, augmenting a Gauss--Markov noise process
into the state casts it in the form of Definition~\ref{def:drift} with a block-structured $A$ (so the
augmented $A$ is not the identity). The two regimes ask different questions: a drifting optimum never
``converges,'' so the natural quantity is the steady-state tracking error $\bar D_\infty(B)$, whereas a
static optimum has the rounds-to-accuracy quantity $T$ of the main text.
\end{remark}

\begin{lemma}[Sequential reduction; query irrelevance]
\label{lem:seqreduce}
Under Definition~\ref{def:drift}: (i) for every $\mathcal F_{t-1}$-measurable query sequence, the joint
law of $\{Y_t\}$ depends only on $\{\theta_t,\xi_t\}$ and not on the queries; and (ii)
$D_t=\tfrac12\E\|\hat\theta_t-\theta_t\|^2$ with $\hat\theta_t$ one-step predictable. Hence bit-constrained
drifting optimization is equivalent to causal source coding of the noisy Gauss--Markov source
$\{Y_t=\theta_t-\xi_t\}$ over a noiseless $B$-bit link with mean-square distortion.
\end{lemma}
\begin{proof}
$Y_t=x_t-g_t=x_t-((x_t-\theta_t)+\xi_t)=\theta_t-\xi_t$ eliminates $x_t$, giving (i); (ii) is the
definition of $D_t$. The main-text Lemma~\ref{lem:reduction} is the case $A=I$, $W=0$.
\end{proof}

\subsection{Lower bound I: a self-contained entropy-power bound}\label{app:seqrd:lb1}
For a random vector $Z\in\R^d$ and a $\sigma$-algebra $\mathcal G$, write the conditional entropy power
(bits) $N(Z\mid\mathcal G):=\tfrac{1}{2\pi e}2^{2h(Z\mid\mathcal G)/d}$; for a Gaussian $\Ncal(\mu,\Sigma)$
one has $N=(\det\Sigma)^{1/d}$.

\begin{lemma}[Rate, observation, time, and distortion bounds]
\label{lem:entropypower}
Let $N_t:=N(\theta_t\mid\mathcal F_{t-1})$ and $N_t^+:=N(\theta_t\mid\mathcal F_t)$. For any encoder, with
no distributional assumption on the (non-Gaussian) posteriors:
\begin{enumerate}[leftmargin=2em,itemsep=1pt]
\item[(R)] $N_t^+\ge N_t\,2^{-2B/d}$, since
$h(\theta_t\mid\mathcal F_{t-1})-h(\theta_t\mid\mathcal F_t)=I(\theta_t;M_t\mid\mathcal F_{t-1})\le H(M_t)\le B$.
\item[(O)] $N_t\ge(\det \bar P^{\mathrm{KF}}_t)^{1/d}=:N^{\mathrm{KF}}_t$, where $\bar P^{\mathrm{KF}}_t$
is the Kalman one-step prediction covariance: since $\mathcal F_{t-1}\subseteq\sigma(Y^{t-1},U)$ and
$\theta_t\mid Y^{t-1}$ is Gaussian, no predictor measurable w.r.t.\ $\mathcal F_{t-1}$ beats the Kalman
prediction error.
\item[(T)] $N_{t+1}\ge|\det A|^{2/d}N_t^+ + \det(W)^{1/d}$, by the entropy-power inequality and
$w_t\perp(\theta_t,\mathcal F_t)$ (Remark~\ref{rem:augment} keeps this independence).
\item[(D)] $D_t\ge\tfrac d2\,N_t$, since $\inf_{\hat\theta_t\in\mathcal F_{t-1}}\E\|\hat\theta_t-\theta_t\|^2=\Tr\bar P_t\ge d(\det\bar P_t)^{1/d}\ge d\,N_t$.
\end{enumerate}
\end{lemma}
\begin{proof}
(R) Differential-entropy reduction equals conditional mutual information, bounded by $H(M_t)\le B$ bits;
exponentiate by $2/d$. (O) Any $\mathcal F_{t-1}$-measurable predictor has error at least the Kalman
prediction error, as $\mathcal F_{t-1}\subseteq\sigma(Y^{t-1},U)$ and $(\theta_t,Y^{t-1})$ are jointly
Gaussian. (T) Entropy-power inequality \citep[Thm.~17.7.3]{cover2006elements} with $N(A\theta_t\mid\cdot)=|\det A|^{2/d}N(\theta_t\mid\cdot)$.
(D) Optimality of the conditional mean, then AM--GM $(\det)^{1/d}\le\Tr/d$ and the maximum-entropy
inequality $N\le(\det)^{1/d}$.
\end{proof}

\begin{theorem}[Self-contained tracking lower bound]
\label{thm:seqlb}
Let $a_\star:=|\det A|^{2/d}2^{-2B/d}$ and $w_\star:=\det(W)^{1/d}$. If $a_\star<1$, equivalently
$B>\log_2|\det A|$, then
\[
\bar D_\infty\ \ge\ \frac d2\,N_\star,\qquad
N_\star=\max\Big\{\underbrace{\tfrac{\det(W)^{1/d}}{1-|\det A|^{2/d}2^{-2B/d}}}_{\text{rate floor}},\ \underbrace{(\det\bar P^{\mathrm{KF}}_\infty)^{1/d}}_{\text{observation floor}}\Big\},
\]
where $\bar P^{\mathrm{KF}}_\infty$ is the stabilizing solution of the Kalman prediction Riccati
$\bar P=A(\bar P^{-1}+V^{-1})^{-1}A^\top+W$ (it exists for any $A$ since $C=I$ is detectable). For stable
$A$ ($|\det A|\le1$, the typical optimization-tracking regime) the rate floor is finite at every
$B\ge0$, so the binding constraint is usually the observation floor. The divergence
$\bar D_\infty\to\infty$ for $B<\log_2|\det A|$, and with it the Nair--Evans data-rate threshold
\citep{nair2004stabilizability}, arises only when $A$ has unstable modes ($|\det A|>1$; e.g.\ a scalar
source with $|a|>1$).
\end{theorem}
\begin{proof}
By (O), $N_t\ge(\det\bar P^{\mathrm{KF}}_t)^{1/d}\to(\det\bar P^{\mathrm{KF}}_\infty)^{1/d}$, so
$\liminf_t N_t\ge(\det\bar P^{\mathrm{KF}}_\infty)^{1/d}$. Composing (R) and (T) of
Lemma~\ref{lem:entropypower} gives $N_{t+1}\ge a_\star N_t+w_\star$; iterating from any $N_0\ge0$,
$N_t\ge a_\star^{t}N_0+w_\star\frac{1-a_\star^{t}}{1-a_\star}$, and when $a_\star<1$ this yields
$\liminf_t N_t\ge w_\star/(1-a_\star)$ (the unique fixed point of the contraction
$\Phi(N)=a_\star N+w_\star$). Combining the two,
$\liminf_t N_t\ge N_\star=\max\{w_\star/(1-a_\star),\,(\det\bar P^{\mathrm{KF}}_\infty)^{1/d}\}$. Finally,
$\bar D_\infty=\limsup_T\frac1T\sum_{t\le T}D_t\ge\frac d2\,\liminf_T\frac1T\sum_{t\le T}N_t\ge\frac d2\liminf_t N_t\ge\frac d2 N_\star$,
using (D) and that a Cesàro mean has $\liminf$ at least that of the sequence. The prediction Riccati has
a unique stabilizing solution $\bar P^{\mathrm{KF}}_\infty\succ0$ because $C=I$ makes $(A,C)$ detectable
\citep{andersonmoore1979}. When $a_\star\ge1$ (i.e.\ $B\le\log_2|\det A|$, which requires
$|\det A|>1$) the recursion has no finite fixed point and the rate floor diverges, recovering the
Nair--Evans threshold.
\end{proof}

\begin{remark}[Honestly loose]
\label{rem:loose}
Theorem~\ref{thm:seqlb} uses one scalar determinant per step, so it inherits the AM--GM gap in (D) and
the entropy-power gap in (T): it is a data-rate-theorem-style bound, not tight in general. The tight
object is the \emph{remote} nonanticipative RDF discussed next.
\end{remark}

\subsection{Lower bound II: the remote nonanticipative RDF}\label{app:seqrd:lb2}
\begin{lemma}[The relevant object is a remote NRDF]
\label{lem:srd}
Because the encoder observes the \emph{noisy} $Y_t=\theta_t-\xi_t$ but distortion is measured against
$\theta_t$, the relevant rate--distortion object is the \emph{remote} (indirect) nonanticipative RDF
$R^{\mathrm{rem}}_{\mathrm{na}}(D)$ for conveying $\theta_t$ from causally encoded noisy observations,
not the ordinary SRD of a directly observed source. It is semidefinite-representable
\citep{tatikonda2009capacity,tanaka2017sdp}; the vector dynamic reverse-waterfilling \emph{closed form}
is in general incorrect \citep{stavrou2020revisited}, so the SDP is the safe route. Two consequences are
worth isolating: (a) the scalar rate--distortion \emph{term} for the underlying Gauss--Markov source is
$\tfrac12\log_2(a^2+W/D)$; and (b) because the observation is noisy, no rate removes the Kalman
filtering floor, so $R^{\mathrm{rem}}_{\mathrm{na}}(D)=\infty$ for $D$ below that floor.
\end{lemma}

\begin{proposition}[Directed-information converse, with the observation floor]
\label{thm:srdconverse}
Directed information satisfies $I(Y^T\!\to M^T)=\sum_t I(Y^t;M_t\mid M^{t-1})\le TB$
\citep{massey1990directed}, hence $\bar D_\infty\ge (R^{\mathrm{rem}}_{\mathrm{na}})^{-1}(B)$. In
particular the distortion cannot fall below \emph{either} the rate floor \emph{or} the observation
(Kalman) floor,
\[
\bar D_\infty\ \ge\ \tfrac d2\,N_\star
=\tfrac d2\max\Big\{\tfrac{\det(W)^{1/d}}{1-|\det A|^{2/d}2^{-2B/d}},\ (\det\bar P^{\mathrm{KF}}_\infty)^{1/d}\Big\},
\]
\emph{exactly the bound of Theorem~\ref{thm:seqlb}} (we use the same prediction-error normalization
$D_t=\tfrac12\E\|x_t-\theta_t\|^2$, with $x_t$ a one-step predictor). The remote (noisy-observation)
aspect is what keeps the Kalman floor present at every rate: even as $B\to\infty$ the decoder recovers
only $Y^t$, so the best one-step predictor of $\theta_t$ from $Y^{t-1}$ still has error
$\tfrac d2(\det\bar P^{\mathrm{KF}}_\infty)^{1/d}>0$. A \emph{direct}-source rate--distortion formula,
which ignores the observation noise and vanishes as $B\to\infty$, would therefore understate the
distortion at high rate; the remote correction is exactly this floor.
\end{proposition}

\subsection{Achievability via innovation quantization (expected-rate)}\label{app:seqrd:ach}
\begin{lemma}[Dithered-lattice innovation quantizer]
\label{lem:innovation}
Let the encoder maintain the Kalman one-step prediction $\hat Y_{t\mid t-1}$ and quantize the innovation
$\nu_t:=Y_t-\hat Y_{t\mid t-1}$ (conditionally Gaussian, covariance $S_t\succ0$) with a subtractive-dither
nested lattice (ECDQ) followed by entropy coding. The reconstruction is unbiased,
$\E[\hat\nu_t\mid\nu_t]=\nu_t$, with $\E\|\hat\nu_t-\nu_t\|^2\le c_d\,2^{-2B/d}\det(S_t)^{1/d}d$
($c_d\le 2\pi e/12$ as $d\to\infty$) at \emph{average} rate $B$ bits per round. The scheme needs only
finite differential entropy of $\nu_t$ (equivalently $S_t\succ0$); no a.s.\ bound on $\|g_t\|$ and no
bounded innovation support is required \citep{zamir1998nested,kostina2019ratecost}.
\end{lemma}

\begin{proposition}[Expected-rate achievability]
\label{thm:seqach}
ECDQ quantization of the innovation achieves \emph{average} rate $B$ bits per round---not a fixed-length
message $M_t\in\{0,1\}^B$ as in Definition~\ref{def:oracle}. In this \emph{expected-rate} model and in
the high-rate regime, the Kalman filter composed with the quantizer of Lemma~\ref{lem:innovation}
attains a steady-state distortion within the space-filling factor $c_d$ of the remote-NRDF lower bound
of Proposition~\ref{thm:srdconverse}, i.e.\ the $\max$ of a $c_d$-inflated rate-limited term and the
Kalman floor. The scheme requires only that the innovation have finite differential entropy (its
\emph{support} is unbounded, as for any Gaussian; only its conditional covariance is finite). It does
not require an a.s.\ bound on $\|g_t\|$. Consequently it addresses the oracle gap for the drifting model
in the \emph{expected-rate} sense; the exact fixed-length per-round gap (L3$'$) remains open.
\end{proposition}
\begin{proof}
With subtractive dither, the ECDQ reconstruction error $\hat\nu_t-\nu_t$ is independent of $\nu_t$ and
has covariance $\preceq c_d\,2^{-2B/d}\det(S_t)^{1/d}I$ (Lemma~\ref{lem:innovation};
\citealp{zamir1998nested}). The decoder thus sees the innovation through an additive independent noise
of that covariance; substituting it into the Kalman recursion gives a perturbed prediction Riccati whose
stabilizing solution is, at high rate (where the lattice operates near the Gaussian rate--distortion
function, the overload probability is negligible, and the entropy-coded \emph{average} rate is $B$),
within $c_d$ of the remote-NRDF distortion of Proposition~\ref{thm:srdconverse}. This is the
innovation-quantization achievability of \citet{kostina2019ratecost}, valid whenever the noise has
finite differential entropy and hence requiring no a.s.\ bound on $g_t$. The matching is in the
average-rate sense; converting it to an exact fixed-length $\{0,1\}^B$ code per round is the open
fixed-length question (L3$'$).
\end{proof}

\begin{remark}
Existing single-machine quantizers \citep[e.g.][]{alistarh2017qsgd} quantize the \emph{raw} gradient
under an a.s.\ norm bound; quantizing the innovation is what simultaneously exploits temporal correlation
and removes that bound. The high-rate space-filling constant $c_d$ and an overload tail (absorbed into
$c_d$) prevent an exactly tight matching constant.
\end{remark}

\subsection{Specialization and the correlated-gradient correction}\label{app:seqrd:correct}
We now treat Model A: a \emph{fixed} optimum $\theta$ with stationary AR(1) gradient noise
$\xi_t=\rho\,\xi_{t-1}+\sqrt{1-\rho^2}\,\eta_t$, $\eta_t\stackrel{\mathrm{iid}}{\sim}\Ncal(0,\sigma^2 I)$,
stationary variance $\sigma^2$, $|\rho|<1$.

We restate the main-text corollary and give the full argument, then generalize it.
\begin{proof}[Proof of Corollary~\ref{cor:correlated}]
By Lemma~\ref{lem:reduction}, $Y_t=x_t-g_t=\theta-\xi_t$. The causal, invertible transform
$\tilde Y_t:=Y_t-\rho Y_{t-1}=(1-\rho)\theta-\sqrt{1-\rho^2}\,\eta_t$ ($t\ge2$) depends only on
$\{\eta_t\}$, so $\{\tilde Y_t\}_{t\ge2}$ is i.i.d.\ and $(Y_1,\tilde Y_2,\dots,\tilde Y_T)$ is a bijective
(hence sufficient) function of $Y^T$, losing no information about $\theta$. Normalizing,
$\tilde Y_t/(1-\rho)\sim\Ncal(\theta,\sigma^2_{\mathrm{eff}}I)$ with
\[
\sigma^2_{\mathrm{eff}}=\frac{1-\rho^2}{(1-\rho)^2}\,\sigma^2=\frac{1+\rho}{1-\rho}\,\sigma^2.
\]
The encoder (which under Definition~\ref{def:drift} sees $Y^t$) can form $\tilde Y_t$ causally and the
decoder can invert; the $B$-bit budget is unchanged, so under this causal-history encoder the two
minimax risks are equal and Theorem~\ref{thm:product} applies with $\sigma^2\mapsto\sigma^2_{\mathrm{eff}}$.
The map $Y^T\mapsto(Y_1,\tilde Y_2,\dots,\tilde Y_T)$ is a bijection, so the $T$-round correlated problem
has the same minimax risk as the $(T\!-\!1)$-sample i.i.d.\ problem at variance
$\sigma^2_{\mathrm{eff}}$ together with the single Gaussian observation $Y_1$ (variance $\sigma^2$); the
latter shifts the per-round risk by $O(1/T)$ and is dominated. Finally, a causal-history encoder is
stronger than the memoryless encoder of Definition~\ref{def:oracle} (Remark~\ref{rem:strong}), so the
lower bound transfers \emph{a fortiori} to the original oracle.
\end{proof}

\begin{remark}[What the correlated-gradient conjecture got wrong]
\label{rem:correction}
Positive noise correlation \emph{raises} the bound by the factor $\tfrac{1+\rho}{1-\rho}\ge1$: it does
not relax the $d/B$ penalty in isolation but scales \emph{all} terms by $\sigma^2_{\mathrm{eff}}/\sigma^2$.
Three independent computations agree: (a) the whitening above; (b) the AR(1) noise spectrum
$S_\xi(\omega)=\tfrac{(1-\rho^2)\sigma^2}{1-2\rho\cos\omega+\rho^2}$ evaluated at the DC component that
carries $\theta$, $S_\xi(0)=\tfrac{1+\rho}{1-\rho}\sigma^2$; and (c) the effective sample size
$\mathrm{Var}(\bar\xi)=\tfrac{\sigma^2}{T}\sum_k\rho^{|k|}=\tfrac{\sigma^2}{T}\tfrac{1+\rho}{1-\rho}$. What
\emph{does} relax the bit requirement is predictability of the optimum's \emph{trajectory} (slow drift,
$\rho(A)\to1$ with small $W$), via the rate floor of Theorem~\ref{thm:seqlb}---not correlation of the
noise. The original conjecture conflated the two.
\end{remark}

\begin{corollary}[General PSD noise; consistency]
\label{cor:psd}
If $\{\xi_t\}$ is stationary with matrix spectral density $S_\xi(\omega)$, estimating the fixed
$\theta$ (the DC component) has effective noise covariance $S_\xi(0)$, and the bound is obtained by
reverse-waterfilling over the eigenvalues of $S_\xi(0)$ as a corollary of the SDP of
Lemma~\ref{lem:srd}. Setting $\rho=0$ (or $A=I$, $W=0$) gives $\sigma^2_{\mathrm{eff}}=\sigma^2$ and
recovers Theorem~\ref{thm:product} exactly.
\end{corollary}

\begin{figure}[t]
\centering
\includegraphics[width=0.92\linewidth]{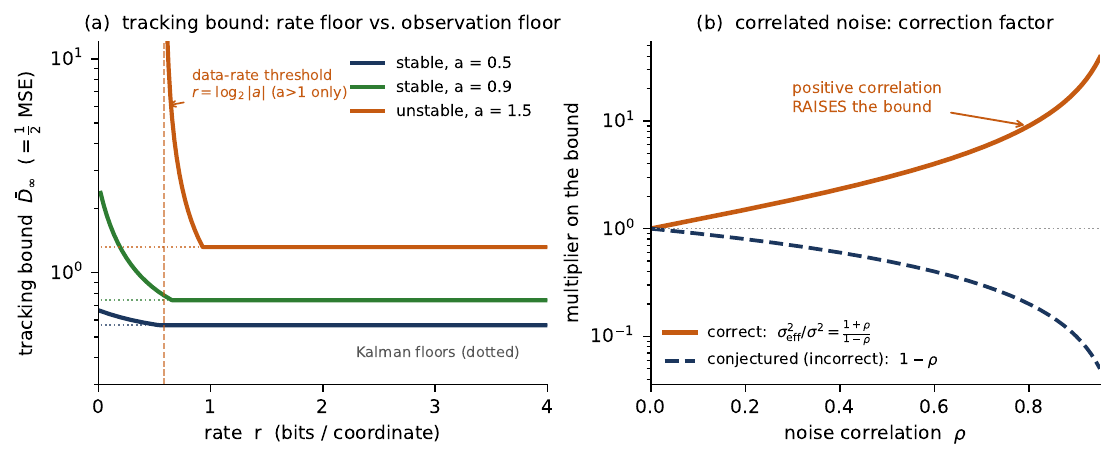}
\caption{Scalar specializations of this appendix (exact evaluations of the closed forms, not
simulations; $W=V=1$, prediction-error normalization $\bar D_\infty=\tfrac12\,$MSE).
\textbf{(a)} The tracking bound of Theorem~\ref{thm:seqlb},
$\bar D_\infty=\tfrac12\max\{W/(1-a^2 2^{-2r}),\,\bar P^{\mathrm{KF}}_\infty\}$, versus per-coordinate
rate $r$. For \emph{stable} $a$ ($|a|<1$) it is finite at every rate, saturating at the
observation-limited Kalman floor (dotted) at high rate. For an \emph{unstable} mode ($a=1.5$) it diverges
at the data-rate threshold $r=\log_2|a|$ (Nair--Evans), confirming that the threshold is an
unstable-mode phenomenon. \textbf{(b)} The correlated-noise correction factor of
Corollary~\ref{cor:correlated}: the effective-variance multiplier $\tfrac{1+\rho}{1-\rho}$ (solid)
\emph{rises} with correlation, correcting the conjectured $1-\rho$ relaxation (dashed).}
\label{fig:seqrd}
\end{figure}

\end{document}